\documentclass{article}

\usepackage{fullpage}

\usepackage[ruled]{algorithm2e} 

\SetAlFnt{\small}
\SetAlCapFnt{\small}
\SetAlCapNameFnt{\small}
\SetAlCapHSkip{0pt}
\IncMargin{-\parindent}

\usepackage{amsmath, amssymb, amsthm}
\usepackage{verbatim}
\usepackage{array}
\newcolumntype{C}[1]{>{\centering\let\newline\\\arraybackslash\hspace{0pt}}m{#1}}
\usepackage{tikz}
\usetikzlibrary{arrows}
\usepackage{hyperref}
\hypersetup{colorlinks=true,urlcolor=blue,linkcolor=blue,citecolor=blue}

\newcommand{\compilehidecomments}{false}

\ifthenelse{ \equal{\compilehidecomments}{true} }{%
	\newcommand{\wei}[1]{}
	\newcommand{\hanrui}[1]{}
	\newcommand{\shanghua}[1]{}
	\newcommand{\vince}[1]{}
}{
	\newcommand{\wei}[1]{{\color{blue!50!black}  [\text{Wei:} #1]}}
	\newcommand{\hanrui}[1]{{\color{brown!60!black} [\text{Hanrui:} #1]}}
	\newcommand{\shanghua}[1]{{\color{magenta} [\text{Shanghua:} #1]}}
	\newcommand{\vince}[1]{{\color{green} [\text{Vince:} #1]}}
}

\DeclareMathOperator*{\argmax}{argmax}

\newcommand{\sm}{\mathrm{SMW}}
\newcommand{\sd}{\mathrm{SD}}
\newcommand{\sa}{\mathrm{SAW}}
\newcommand{\mph}{\mathrm{MPH}}
\newcommand{\ph}{\mathrm{PH}}
\newcommand{\supadd}{\mathrm{SUPADD}}

\newtheorem{definition}{Definition}
\newtheorem{theorem}{Theorem}
\newtheorem{lemma}{Lemma}[section]
\newtheorem{proposition}{Proposition}[section]
\newtheorem{corollary}{Corollary}[section]

\begin{document}

\title{Capturing Complementarity in Set Functions by Going Beyond Submodularity/Subadditivity}  

\author{Wei Chen \\ Microsoft Research \\ \texttt{weic@microsoft.com}
\and Shang-Hua Teng \\ USC \\ \texttt{shanghua@usc.edu}
\and Hanrui Zhang \\ Duke University \\ \texttt{hrzhang@cs.duke.edu}}


\maketitle


\begin{abstract}
We introduce two new ``degree of complementarity'' measures,
   which we refer to, respectively, as {\em supermodular width}
   and {\em superadditive width}.
Both are formulated based on natural witnesses of complementarity.
We show that both measures are robust by proving 
  that they, respectively, characterize
   the gap
   of monotone set functions from being submodular and subadditive.
Thus, they define two new hierarchies over monotone set functions, 
  which we will refer to as 
  Supermodular Width (SMW) hierarchy and Superadditive 
  Width (SAW) hierarchy,
  with foundations --- i.e. level 0 of the hierarchies --- 
  resting exactly on
  submodular and subadditive functions, respectively.

We present a comprehensive comparative analysis
  of the SMW hierarchy 
  and the Supermodular Degree (SD) hierarchy, 
  defined by Feige and Izsak.
We prove that the SMW hierarchy 
  is strictly more expressive than the SD hierarchy.
In particular, we show that every monotone set function of supermodular 
  degree $d$ has supermodular width at most $d$,
  and there exists a supermodular-width-1 function
  over a ground set of $m$ elements
  whose supermodular degree is $m-1$.
We show that previous results 
  regarding approximation guarantees for
  welfare and constrained maximization as well as 
    regarding the Price of Anarchy (PoA) of simple auctions
  can be extended without any loss from the supermodular degree
  to the supermodular width.
We also establish almost matching 
  information-theoretical lower bounds
  for these two well-studied fundamental maximization problems
  over set functions.
The combination of these approximation and hardness results 
  illustrate that the $\sm$ hierarchy 
  provides not only a natural notion of complementarity,
  but also an accurate characterization 
  of ``near submodularity'' needed 
  for maximization approximation.
While $\sd$ and $\sm$ hierarchies support
  nontrivial bounds on the PoA of 
  simple auctions, we show that our $\sa$ hierarchy 
  seems to capture more intrinsic properties needed
  to realize the efficiency of simple auctions. 
So far, the $\sa$ hierarchy provides the best dependency for the PoA of Single-bid
  Auction, and is nearly as competitive as 
  the Maximum over Positive Hypergraphs (MPH) hierarchy 
  for Simultaneous Item First Price Auction (SIA).
We also provide almost tight lower bounds
  for the PoA of both auctions with respect to the $\sa$ hierarchy.

\end{abstract}


\section{Introduction}

\newcommand{\Real}{\mathbb{R}}

For a ground set $X = [m] = \{1,2,\dots, m\}$, 
  a set function $f: 2^X \rightarrow \Real$
  assigns each subset $S\subseteq X$ a real value.\footnote{Throughout the paper we use $m$ to denote the number of elements in the ground set.}
Function $f$ is {\em monotone}
  if $f(T) \geq f(S), \forall S\subseteq T\subseteq X$,
   and {\em normalized} if $f(\emptyset) = 0$.
In this paper, we will focus on normalized monotone set functions,
  which by definition are non-negative.

Like graphs to network analysis, set functions 
  provide the mathematical language
  for many applications, ranging from 
  combinatorial auctions (economics)
  to coalition formation (cooperative game theory; political science) 
  \cite{ScarfCore,Shapley}
  to influence maximization (viral marketing) 
   \cite{RichardsonDomingos,KKT}.
Because of its exponential dimensionality,
  set functions --- which are as rich as weighted hypergraphs 
  --- are far more expressive mathematically and 
  challenging algorithmically 
  than graphs \cite{JMPageRankEssence}.
However, when monotone set functions are 
  {\em submodular} \cite{nemhauser1978analysis,vondrak2008optimal}
  or ---  more generally --- {\em complement-free}
   \cite{feige2009maximizing},
  algorithms with remarkable performance guarantees have been 
  developed for various optimization, social influence,
  economic, and learning tasks
  \cite{balcan2011learning,KKT,mirrokni2008tight,Chawla,Nisan}.

In this paper, 
  we study two new {\em degree-of-complementarity} measures
  of monotone set functions, and demonstrate 
  their usefulness for several optimization and economic tasks.
We prove that our complementarity measures
  --- which are based on natural {\em witnesses} of complementarity --- 
  introduce hierarchies (over monotone set functions)
  that smoothly move beyond submodularity and subadditivity.

\subsection{Witnesses to Complementarity: Supermodular Sets and Superadditive
Sets}  

For any sets $S, T \subseteq X$, 
  let $f(S | T) := f(S \cup T) - f(T)$ be the {\em margin} of $S$ given $T$.
Recall that $f$ is {\em subadditive} if 
  $f(S \cup T) \leq f(S) + f(T),\, \forall S, T \subseteq X$,
 and {\em submodular} if for all $S, T$ and $v \in X \setminus T$, 
  $f(v | S \cup T) \leq f(v | S)$.
It is well known that every submodular set function is also subadditive.

If there are sets $S, T \subseteq V$ such that $f(S \cup T) > f(S)+ f(T)$,
  then one may say that $(S,T)$ is a witness to {\em complementarity} 
  in monotone set function $f$.
Motivated by a line of recent work 
  \cite{abraham2012combinatorial,feldman2014constrained,feige2013welfare,feige2015unifying,feldman2016simple,eden2017simple},
  we consider the following fundamental question about set functions:

\begin{quote}
{\em Are there other natural, and preferably more general, 
  forms of witnesses to complementarity that have algorithmic consequences?}
\end{quote}

The supermodular degree of Feige and Izsak \cite{feige2013welfare}
  is among the first measures of complementarity that
  are connected with algorithmic solutions to 
  monotone-set-function maximization and combinatorial auctions.
The supermodular degree is defined based on a notion of
 positive dependency between elements:
 $u \in X$ positively depends on $v \in X \setminus \{u\}$ 
 (denoted by $u \rightarrow^+ v$), 
  if there exists $S \subseteq X$ 
  such that $f(u | S) > f(u | S \setminus \{v\})$. 


\begin{definition}[Supermodular Degree] 
\label{def:supdegree}
The supermodular degree of a set function 
  $f: 2^X \rightarrow \mathbb{R}^+$, $\sd(f)$, 
  is defined to be $\sd(f) = \max_{u \in X} |\mathrm{Dep}_f^+(u)|$,
  where $\mathrm{Dep}_f^+(u) = \{v | u \rightarrow^+ v\}$.
\end{definition}

Although supermodular degree has been shown useful in a number of settings, it is not clear whether it provides the tightest possible characterization of supermodularity.
For example, consider a customer who wants any two or more items out of $m$ items, but not zero or one item.
That is, the customer has a valuation function, where any subset of $[m]$ of size at least $2$ provides utility $1$, and any subset of size at most $1$ provides utility $0$.
For this function, according to Feige and Izsak's definition, any two items depend positively on each other.
In particular, any item depends positively on all other items, so the supermodular degree of this valuation function is $m - 1$ --- the largest degree possible.
This seems to contradict the intuition that there is only very limited complementarity.

Below, we will provide two perspectives, with the first
   highlighting {\em supermodularity} and the second
   highlighting {\em superadditivity}.
In the rest of the paper, we will study how these two complementarity measures 
  can be used to 
  capture the performance of basic computational solutions
   in optimization and auction settings where the utilities
   are modeled by monotone set functions.
In particular, our measure of supermodularity refines supermodular degree, and avoids the kind of overestimation discussed above.

Our first definition focuses on modularity:

\begin{definition}[Supermodular Set]\label{def:supermodularSet}
Given a normalized monotone set function $f$ over a ground set $X$, 
  a set $T \subseteq X$ is {\em supermodular} w.r.t.\ $f$ if 
\[
    \exists\, S \subseteq X \mbox{{ \rm and }} v \in X \setminus T,\, \text{such that:}\ f(v | S \cup T) > \max_{T' \subsetneq T} f(v | S \cup T').
\]
\end{definition}

First note that if $f$ is submodular, then 
  $f(v| S\cup T) \leq f(v| S\cup T'), \forall T' \subsetneq T$,
  implying $f$ has no supermodular set.
Thus, if a set function $f$ has a supermodular set, then it is not submodular.

We say that such a set $T$ (in Definition \ref{def:supermodularSet})
  {\em complements} item $v$ {\em given} $S$.
In other words, $S$ provides the setting that demonstrates
  the complementarity between $v$ and $T$.
In the customer example given after Definition~\ref{def:supdegree}, we can easily check that any singleton
is a supermodular set, but any set with size at least two is not a supermodular set, because any single item in
the set already provides all the complementarity for any other single item.
A supermodular set behaves similarly 
  to the typical example of complements, namely {\em complementary bundles},\footnote{$S$ is a {\em complementary bundle} if $f(S) > 0$ and $\max_{S' \subsetneq S} f(S') = 0$.} 
  in the sense that the set as a whole provides 
  more complement to a single item than any of its strict subsets.
However, supermodular sets have 
  richer structures while preserving the strong complementarity of such bundles,
  making them potentially more challenging to 
  deal with mathematically/algorithmically
  than complementary bundles of a similar size.

Our next definition focuses on additivity: 
\begin{definition}[Superadditive Set]\label{def:SuperAdditiveSet}
Given a normalized monotone set function $f$ over a ground set $X$, 
  a set $T \subseteq X$ is superadditive w.r.t.\ $f$ if
\[
    \exists\, S \subseteq X \setminus T \text{ such that:}\ f(S | T) > \max_{T' \subsetneq T} f(S | T').
\]
\end{definition}
In Definition \ref{def:SuperAdditiveSet}, we
    say such a set $T$ {\em complements} set $S$.
Note that if $f$ is subadditive, then for $T' = \emptyset$,
  $f(S| T) = f(S\cup T) - f(T) \leq (f(S) + f(T)) - f(T) = 
  f(S) = f(S) - f(T') = f(S|T')$, implying $f$ does not have a superadditive
  set.
In other words, if $f$ has any superadditive set, then it is not subadditive.

Supermodular/superadditive sets correspond to witnesses
   that exhibit different kinds of complementarity.
Supermodular sets are sensitive to the presence of 
  an environment, and superadditive sets model 
  complements to sets instead of items.
The cardinality of the largest supermodular sets or
  superadditive sets provides a measure of 
  the ``level of complementarity'', 
  similar to the supermodular degree (\cite{feige2013welfare}), 
      the size of the largest bundle, and 
      the hyperedge size (\cite{feige2015unifying}) (also see Definition \ref{MPH}) 
      in previous work.

\begin{definition}[Supermodular Width]
    The {\em supermodular width} of a set function $f$ is defined to be
    \[
        \sm(f) := \max \{|T| \mid T \text{ is a supermodular set w.r.t.\ } f\}.
    \]
\end{definition}

\begin{definition}[Superadditive Width]
    The {\em superadditive width} of a set function $f$ is defined to be
    \[
        \sa(f) := \max \{|T| \mid T \text{ is a superadditive set w.r.t.\ } f\}.
    \]
\end{definition}


Each measure classifies monotone set functions into a hierarchy of
  $m$ levels:
\begin{definition}[Supermodular Width Hierarchy ($\sm\text{-}d$)]
For any integer $d \in \{0, \dots, m - 1\}$, 
  a set function $f: 2^{[m]} \rightarrow \mathbb{R}$
  belongs to the first $d$-levels of the {\em supermodular width hierarchy},
   denoted by $f \in \sm\text{-}d$, if and only if $\sm(f) \le d$.
\end{definition}
\begin{definition}[Superadditive Width Hierarchy ($\sa\text{-}d$)]
For any integer $d \in \{0, \dots, m - 1\}$, 
a set function $f: 2^{[m]} \rightarrow \mathbb{R}$ belongs to the first $d$ levels 
  of the {\em superadditive width hierarchy}, denoted by $f \in \sa\text{-}d$, 
  if and only if $\sa(f) \le d$.
\end{definition}

We will show that functions at level 0 of 
 the above two hierarchies, respectively, are precisely 
 the families of submodular and subadditive functions.
In both hierarchies, $\sm\text{-}(m-1)$ and $\sa\text{-}(m-1)$
  contains all monotone set functions over $m$ elements.
Coming back to the customer example again, we see that the utility of the customer has supermodular width of $1$.
Comparing to its supermodular degree of $m-1$, our hierarchy characterizes this utility function at a much lower
	level, which matches our intuition that the complementarity of this customer's utility function should be limited.
We will further show below that this difference would also have significant algorithmic implications.

\subsection{Our Results and Related Work}

We now summarize the technical results of this paper.
Structurally, we provide strong evidence that our definitions
  of supermodular/superadditive sets are natural and robust.
We show that they ---  respectively --- capture a set-theoretical gap
  of monotone set functions to submodularity and subadditivity.
Algorithmically, we prove that 
  our characterization based on supermodular width is {\em strictly stronger}
  than that of Feige-Izsak's  based on supermodular degree,
  by establishing the following:
\begin{enumerate}
\item  For every set function $f: 2^{[m]}\rightarrow \mathbb{R}$, 
$\sd(f)\leq \sm(f)$, and there exists a function whose supermodular degree is much 
  larger than its supermodular width.
\item The  $\sm$ hierarchy
   offers the same level of algorithmic guarantees in 
   the maximization and auction settings as the $\sd$ hierarchy.
\end{enumerate}
We will also compare both hierarchies with the $\mph$ hierarchy of 
  \cite{feige2015unifying}.

\subsubsection{Robustness: Capturing the 
   Set-Theoretical Gap to Submodularity/Subadditivity}

We interpret the level of complementarity
  in our formulation of supermodular and superadditive sets
  from a dual perspective:
We prove that they characterize the gaps from a monotone set function to 
  submodularity and subadditivity, respectively.
Our characterization uses the following definition.

\begin{definition}[$d$-Scopic Submodularity]
For integer $d \geq 0$, a normalized monotone set function $f$ is 
   $d$-scopic submodular if and only if,
\begin{eqnarray}\label{ScopicSubmodular}
f(v|T) \leq \max_{T': T' \subseteq T, |T'| \leq d} f(v | S \cup T'), 
 \quad \forall S, T \subseteq X, v \in X \text{ satisfying } S \subseteq T,\, v \notin T
\end{eqnarray}
\end{definition}

Note that in Condition~\eqref{ScopicSubmodular},
  the family $\{S\cup T' | T' \subseteq T, |T'|\leq d\}$
  defines a set-theoretical {\em neighborhood} around $S$.
Our definition of $d$-scopic submodularity means that 
  even if the submodular condition $f(v|T) \leq f(v|S)$ may not hold 
  for some $S \subseteq T$, it holds for some set in 
  $S$'s $d$-neighborhood  inside $T$.
Thus, the parameter $d$ provides a set-theoretical scope for examining 
  submodularity.

Similarly, we define:
\begin{definition}[$d$-scopic Subadditivity]
For integer $d\geq 0$, a set function $f$ 
  is $d$-scopic subadditive if and only if,
\begin{eqnarray}\label{ScopicSubadditivity}
      f(S | T) \le \max_{T': T' \subseteq T,\, |T'| \le d} f(S | T'), 
  \quad \forall S, T \subseteq X \text{ satisfying } S \cap T = \emptyset.
\end{eqnarray}
\end{definition}



In Section~\ref{sec:Expressiveness}, we prove the following two 
  theorems.

\begin{theorem}[Set-Theoretical Characterization of the $\sm$ Hierarchy]\label{Robustness}
For any integer $d \geq 0$ and set function $f: 2^X\rightarrow \mathbb{R}$,
  $f$ is  $d$-scopic submodular if and only if $\sm(f) \le d$.
\end{theorem}
\begin{theorem}[Set-Theoretical Characterization of the $\sa$ Hierarchy]
\label{RobustnessAdditive}
For any integer $d \geq 0$ and set function $f: 2^X\rightarrow \mathbb{R}$,
  $f$ is  $d$-scopic subadditive if and only if $\sa(f) \le d$.
\end{theorem}

With matching supermodularity/submodularity and superadditivity/subadditivity
   characterization, 
   Theorems~\ref{Robustness}~and~\ref{RobustnessAdditive} 
  illustrate that our definitions of
    supermodular/superadditive sets are both natural and robust.  
We note that while monotone submodular functions are all subadditive,
  some $d$-scopic submodular functions are not
  $d$-scopic subadditive.
As shown in
   Propositions~\ref{prop:sa_much_larger_than_sm} and
   \ref{prop:sm_much_larger_than_sa}, 
   these two hierarchies are not comparable.
We will show that they model different aspects of
  complementarity that can be utilized in different
  algorithmic and economic settings.


\subsubsection{Expressiveness: Strengthening Supermodular Degree}

We will show that our characterization based on supermodular width strengthens
  Feige-Izsak's
  the characterization based on supermodular degree \cite{feige2013welfare}.
The statement has two parts. 
We first prove, that supermodular sets extend
  positive dependency (as used in supermodular degree), which ---  in essence --- 
   can be viewed as a graphical approximation of supermodular sets.
 
\begin{theorem}\label{Theo:SupermodularDegree}
Every monotone set function $f$ with 
  supermodular degree $d$ has supermodular width at most $d$
 (i.e., it is $d$-scopic submodular).
Moreover, there exists a monotone set function $f: 2^{[m]}\rightarrow \mathbb{R}^+$ with 
$\sm(f) = 1$ and $\sd(f) = m-1$.
\end{theorem}

In other words, 
  the $\sm$ hierarchy strictly {\em dominates} 
  the $\sd$ hierarchy.
\footnote{Formally, when 
  comparing two set-function hierarchies, say with name $\{\text{Y}_d\}_{d\in [0,m-1]}$
   and $\{\text{Z}_d\}_{d\in [0,m-1]}$,  
  we say Y {\em dominates} 
  Z, 
  if for all $d\in [0,m-1]$ and $f$, $f\in$ $\text{Z}_d$ implies $f\in$ $\text{Y}_d$.}



\subsubsection{Usefulness: Algorithmic and Economic Applications}

We then show, algorithmically,
  the $\sm$ hierarchy ---  while being more expressive 
   than the $\sd$ hierarchy ---
   is almost as useful as the latter (Theorems~\ref{thm:constrained_maximization},~\ref{thm:welfare_maximization}~and~\ref{thm:sia_with_sm}). 






We will illustrate the usefulness of 
  our hierarchies in algorithm and auction design
   with two archetypal classes of problems,
   {\em set function maximization} and {\em combinatorial auctions},
  which traditionally involve measures of complementarity.
Motivated by previous 
 work \cite{feige2013welfare,feldman2014constrained,feldman2016simple,feige2015unifying},
 we will characterize the {\em approximation guarantee} 
   of polynomial-time set-function maximization algorithms 
   and {\em efficiency} of simple auction protocols in terms
   of the complementarity level in our hierarchies.
In these settings, we will compare our hierarchies with 
  two most commonly cited complementarity hierarchies: 
  the supermodular degree ($\sd$) hierarchy and 
  the Maximum over Positive Hypergraphs ($\mph$) hierarchy. 

\begin{itemize}
\item {{\em Set-Function Maximization}} We will consider both
  constrained and welfare maximization.
The former aims to find a set of a given cardinality with maximum function value.
The latter aims to allocate a set of items to $n$ agents, 
  \footnote{Throughout the paper we use
 $n$ to denote the number of agents 
  (whenever applicable) unless otherwise specified.
}
  with potentially different valuations, such that the total 
  value of all agents is maximized.
As a set function has exponential dimensions in $m$, 
  in both maximization problems, we assume that the input set functions
  are given by their value oracles.

\item {{\em Combinatorial Auctions and Simple Auction Protocols}}
We will consider two well-studied simple combinatorial auction protocols: 
  Single-bid Auction and Simultaneous Item First Price Auction (SIA). 
In both settings, there are multiple agents, each of which has
  a (potentially different) valuation function over subsets of items.
The former auction protocol proceeds by asking each bidder to bid a single price, 
  and letting bidders, in descending order of their bids, 
  buy any available set of items paying their bid for each item. 
The latter simply runs first-price auctions simultaneously for
  all items. 
\end{itemize}

\subsubsection*{{\sc Approximation Guarantees According to Supermodular Widths}}

We will prove that the elegant approximability results
  for constrained maximization by \cite{feldman2014constrained}
  and for welfare maximization by \cite{feige2013welfare} 
  can be extended from supermodular degree to supermodular width.
We obtain the same dependency 
  (see Theorems~\ref{thm:constrained_maximization}~and~\ref{thm:welfare_maximization}) 
  --- 
   that is, $1 - e^{-1/(d + 1)}$ and $\frac{1}{d + 2}$ respectively ---
   on the supermodular width $d$ as what the supermodular degree 
   previously provides for these problems.


Because our $\sm$ hierarchy is strictly more expressive,
  our upper bounds for $\sm\text{-}d$ 
  cover strictly more monotone set functions than previous results for $\sd\text{-}d$.
We will also complement our algorithmic results with nearly matching 
  information theoretical lower bounds
  (see Theorems~\ref{theo:CMLowerBound}~and~\ref{theo:WMLowerBound}),
  for these two well-studied fundamental maximization problems.
Our approximation and hardness results 
  illustrate that the $\sm$ hierarchy 
  not only captures a natural notion of complementarity,
  but also provides an accurate characterization 
  of the ``nearly submodular property''  needed 
  by approximate maximization problems.



\subsubsection*{{\sc Efficiency of Simple Auctions According 
  to Superadditive/Supermodular Width}}

Next, we will analyze the efficiency 
  of two well-known simple auction protocols 
  in terms of superadditive width.
To state our results and compare them with previous work, 
  we first recall a notation from \cite{feige2015unifying}:

\begin{definition}[Closure under Maximization]
For any family of set functions $\mathcal{F}$ over $X$, 
  the closure of $\mathcal{F}$ under 
   maximization, denoted by $\max(\mathcal{F})$, is the following family 
   of set functions:
$f \in \max(\mathcal{F})$ if and only if
    \[
        \exists k \in \mathbb{N},\, f_1, \dots, f_k \in \mathcal{F}, \text{ s.t. } f(S) = \max_{i \in [k]} f_i(S), \forall S\subseteq X.
    \]
\end{definition}

We will prove the following:

\begin{theorem}
Single-bid Auction and SIA are approximately efficient
  --- with Price of Anarchy (PoA) $O(d \log m)$ --- for valuation functions
   in $\max(\sa(d))$.
\end{theorem}

We will also complement our PoA results by 
  almost tight (up to a factor of $O(\log m)$) lower bounds:

\begin{theorem}
For any $d>0$, 
  there is an instance with $\sa\text{-}d$ valuations, 
  where the Price of Stability (PoS) of Single-bid Auction is at 
  least $d + 1 - \varepsilon$ for any $\varepsilon > 0$, 
  and the PoA of SIA is at least $d$.
\end{theorem}

Although supermodular width strictly strengthens supermodular degree, 
  superadditive width is not comparable with supermodular degree.
Nevertheless, our PoA bound of $O(d \log m)$
  is a factor of $d$ tighter than the $O(d^2 \log m)$
  supermodular-degree based bound of 
  \cite{feldman2016simple} for Single-bid Auction.
This improvement of dependency on $d$, together with 
  the nearly matching lower bound,
  suggests that the $\sa$ hierarchy might be more capable
  in capturing the smooth transition of efficiency
  of simple auctions.
Furthermore, as a byproduct of our efficiency
  results for the $\sa$ hierarchy, 
  we also obtain similar results, but with a worse dependency on $d$,
  for the $\sm$ hierarchy.

\begin{theorem}
Single-bid auction and SIA are approximately efficient
 --- with PoA $O(d^2 \log m)$ --- for valuations 
in $\max(\sm\text{-}d \cap \supadd)$, 
  where $\supadd$ denotes the class of monotone superadditive set functions.
\end{theorem}


For Single-bid Auction, this result strengthens
  the central efficiency result of \cite{feldman2016simple}
  by replacing the supermodular degree with 
  the more inclusive supermodular width.
For the PoA analysis of SIA, the 
  the Maximum over Positive Hypergraphs (MPH) hierarchy 
  of \cite{feige2015unifying} remains the gold standard, by providing
  asymptotically matching upper and lower bounds.
MPH is defined based on the following hypergraph characterization of
  set functions:
Every normalized monotone set function
  over ground set $X$ can be uniquely expressed
  by another set function $h$ such that 
   $f(S) = \sum_{T \subseteq S} h(T), \forall S \subseteq X$,
  where $h(T)$ for each $T$ is called the weight of hyperedge $T$.

\begin{definition}[Maximum over Positive Hypergraphs \cite{feige2015unifying}]
\label{MPH}
    Let $\ph\text{-}d$ be the class of set functions whose hypergraph representation $h$ satisfies: (1) $h(S) \ge 0$ for all $S$, and (2) $h(S) > 0$ only if $|S| \le d$. The $d$-th level of the MPH hierarchy is defined as $\mph\text{-}d = \max(\ph\text{-}d)$.
\end{definition}

\begin{table}[t]
    \centering
    \begin{tabular}{|C{80pt}|C{65pt}|C{65pt}|C{65pt}|C{65pt}|}
    \hline
    & $\sd\text{-}d$ & $\mph\text{-}(d + 1)$ & $\sm\text{-}d$ & $\sa\text{-}d$ \\
    \hline
    constrained maximization & $1 - e^{1 / (d + 1)}$ \cite{feldman2014constrained} & ? & $1 - e^{1 / (d + 1)}$ (Thm~\ref{thm:constrained_maximization}) & ? \\
    \hline
    welfare maximization & $1 / (d + 2)$ \cite{feige2013welfare} & $1 / (d + 2)$ \cite{feige2015unifying} & $1 / (d + 2)$ (Thm~\ref{thm:welfare_maximization}) & ? \\
    \hline
    PoA of Single-bid Auction & $O(d^2 \log m)$ \cite{feldman2016simple} & ? & $O(d^2 \log m)$ (Thm~\ref{thm:single-bid_with_sm}) & $O(d \log m)$ (Thm~\ref{thm:single-bid_with_sa}) \\
    \hline
    PoA of SIA & $O(d)$ \cite{feige2015unifying} & $O(d)$ \cite{feige2015unifying} & $O(d^2 \log m)$ (Thm~\ref{thm:sia_with_sm}) & $O(d \log m)$ (Thm~\ref{thm:sia_with_sa}) \\
    \hline
    \end{tabular}
    \caption{Comparison of hierarchies of complementarity. Note that the $O(d)$ bound for PoA of SIA with $\sd\text{-}d$ valuations follows from the fact that $\sd\text{-}d \subseteq \mph\text{-}(d + 1)$, which is not clearly comparable with the PoA bound of SIA with $\sm\text{-}d$ valuations. See corresponding references and theorems for more accurate statements.}
    \label{tab:comparison}
\end{table}

MPH provides the best characterization 
  to the efficiency of SIA as well as
  ties with $\sd$ and $\sm$ regarding the approximation
  ratio of welfare maximization (although it requires 
  access to the much stronger demand oracles).
However, it remains open whether it can be used to
  analyze constrained set function maximization and Single-bid Auction.
See Table~\ref{tab:comparison} for a comparison.

We will prove the following theorem which states that, in general,
  the $\sa$ hierarchy is not comparable with MPH.

\begin{theorem}
There is a subadditive function 
  that lives in an upper (i.e.\ $\geq m/2$) 
  $\mph$ level.
On the other direction, there
  is a function on level $2$ of $\mph$,
  whose superadditive and supermodular
  widths are both $m-1$. 
\end{theorem}

It remains open whether $\mph$-$(d + 1)$ --- which subsumes $\sd\text{-}d$
  as a subset  ---  contains $\sm\text{-}d$.
In particular, the proof that $\sd\text{-}d \subseteq \mph\text{-}(d + 1)$
  in \cite{feige2015unifying} does not appear easily applicable
  to $\sm\text{-}d$.


\subsubsection{Other Related Work}

\subsubsection*{{\sc Set Function Maximization}}
There is a rich body of research focusing on set function maximization with complement-free functions, e.g.\ \cite{nemhauser1978analysis, vondrak2008optimal, feige2009maximizing}. 
Various information/complexity theoretical lower bounds have been established for both problems, e.g.\ \cite{nemhauser1978best, feige1998threshold, mirrokni2008tight, khot2005inapproximability}.

\subsubsection*{{\sc Efficiency of Simple Auctions}}
 Single-bid Auction with subadditive valuations has a PoA of $O(\log m)$ \cite{devanur2015simple}. SIA with subadditive valuations has a constant PoA \cite{feldman2013simultaneous}. Posted price auctions with XOS valuations give a constant factor approximate welfare guarantee \cite{feldman2015combinatorial}.

\subsubsection*{\sc {Other Measures of Complementarity}}

Some other useful measures include Positive Hypergraph (PH) \cite{abraham2012combinatorial} and Positive Lower Envelop (PLE) \cite{feige2015unifying}. Eden {\em et al.}\ recently introduce an extensive measure 
  which ranges from $1$ to $2^m$ to 
  capture the smooth transition of 
  revenue approximation guarantee \cite{eden2017simple}.



\section{Expressiveness of the New Hierarchies}
\label{sec:Expressiveness}

\subsection{Characterization of Supermodular/Superadditive Widths}

We first prove Theorems~\ref{Robustness}~and~
  \ref{RobustnessAdditive}, which characterize
  supermodular/superadditive widths with $d$-scopic 
  submodular/subadditive functions.


\begin{proof}[Proof of Theorem~\ref{Robustness}] 
We now show $\sm(f) \le d$ iff $f$ is $d$-scopic submodular.  
First, suppose $\sm(f) \le d$.  
Consider any triple  $(T, S, v)$ such that
   $S \subseteq T \subseteq X$ and $v \not\in T$. 
To show $f$ is $d$-scopic submodular, we prove 
    by induction on the size of $T$, that
\begin{eqnarray}\label{eqn:scopicClaim}
f(v|T) \leq \max_{T': T' \subseteq T, |T'| \leq d} f(v | S \cup T').
\end{eqnarray}

As the base case, when $|T| \le d$, the 
  inequality of (\ref{eqn:scopicClaim}) trivially holds because 
  if $T'= T\setminus S$, then $|T'| \leq d$ and $f(v|S\cup T') = f(v|T)$. 
Inductively, assume that  
  the statement is true for all  $\{V\subseteq X: |V| \le k\}$ 
  for some $k \ge d$. 
Now consider any set $T$ with $|T| = k + 1 > d$.
Because $T$ is not supermodular, there is $T'' \subsetneq T$, 
  such that $f(v | T) \le f(v | T'')$. 
Applying the inductive hypothesis on $(T'',S,v)$, we have:
\[f(v | T'') \le \max_{T':  T' \subseteq T'',\, |T'| \le d } f(v | S \cup T')
\le \max_{T':  T' \subseteq T,\, |T'| \le d } f(v | S \cup T').\]
Thus,   $f(v | T) \le f(v | T'') \le \max_{T':
  T' \subseteq T,\, |T'| \le d } f(v | S \cup T')$,
  and we have demonstrated that $f$ is $d$-scopic submodular.

For the other direction, we assume $f$ is $d$-scopic submodular. 
There is no supermodular set of size larger than $d$, 
  because for any $S$, $T$, $v \notin T$ where $|T| > d$, 
  there is some $T' \subseteq T$ where $|T'| \le d$, 
  such that $f(v | S \cup T) \le f(v | S \cup T')$, 
  i.e.\ $T$ is not supermodular. 
Therefore $\sm(f) \le d$.
\end{proof}

\begin{corollary}
Set function $f$ is submodular iff $\sm(f) = 0$ 
(i.e., $f$ has no supermodular set).
\end{corollary}

\begin{proof}[Proof of Theorem~\ref{RobustnessAdditive}]
We prove $\sa(f) \le d$ iff $f$ is $d$-scopic subadditive.
    Suppose $\sa(f) \le d$. Consider $S$ and $T$ where $S \cap T = \emptyset$. We show $d$-scopic subadditivity by induction on the size of $T$. When $|T| \le d$, the statement trivially holds. Suppose $d$-scopic subadditivity holds for $|T| \le k$ where $k \ge d$. For $|T| = k + 1 > d$, since $T$ is not superadditive, there is $T'' \subsetneq T$, such that $f(S | T) \le f(S | T'')$. Applying inductive hypothesis on $S, T''$ gives $f(S | T) \le f(S | T'') \le \max_{T': T' \subseteq T,\, |T'| \le d} f(S | T')$, i.e.\ $f$ is $d$-scopic subadditive.

Now assume $d$-scopic subadditivity. There is no superadditive set with size larger than $d$, because for any $S$ and $T$ where $|T| > d$ and $S \cap T = \emptyset$, there is some $T' \subseteq T$ where $|T'| \le d$, such that $f(S | T) \le f(S | T')$, i.e.\ $T$ is not superadditive.
\end{proof}

\begin{corollary}
A set function $f$ is subadditive iff $\sa(f) = 0$
 (i.e., $f$ has no superadditive set).
\end{corollary}

\subsection{Supermodular Width vs Supermodular Degree}

The following two propositions establish 
  Theorem~\ref{Theo:SupermodularDegree}, 
  showing supermodular width strictly dominates supermodular degree.

\begin{proposition}
\label{prop:sd_no_smaller_than_sm}
    For any set function $f$, $\sd(f)\leq \sm(f)$.
\end{proposition}
\begin{proof}
    Fix $f$. Let $T$ be a supermodular set of size $\sm(f)$, and $S$, $v$ be such that $f(v | T \cup S) > f(v | T' \cup S),$  
$\forall T' \subsetneq T$. Clearly for any $t \in T$,
        $f(v | \{t\} \cup (T \setminus \{t\}) \cup S) > f(v | (T \setminus \{t\}) \cup S)$.
    In other words, $v \rightarrow^+ t$ for all $t \in T$, so $\sd(f) \ge \mathit{deg}^+(v) \ge |T| = \sm(f)$.
\end{proof}

\begin{proposition}
\label{prop:sd_much_larger_than_sm}
There exists a monotone set function $f$ with 
  $\sm(f) = 1$ and $\sd(f) = m - 1$.
\end{proposition}
\begin{proof}
Consider a symmetric\footnote{$f$ is symmetric if $f(S)$ depends only on $|S|$.} $f$  
  over a ground set $X =[m]$,
   where $f(S) = 0$ if $|S| \le 1$, and $f(S) = 1$ otherwise. 
Observe that for any $u \ne v$, 
  $1 = f(u | \{v\}) > f(u | \emptyset) = f(u) = 0$, 
  so $u \rightarrow^+ v$, and $\sd(f) = |\mathrm{Dep}_f^+(u)| = m - 1$. 
On the other hand, consider any $T$ where $|T| \ge 2$. For any $v$, $S$, 
  we have $|S \cup T| \ge 2$, so
        $0 = f(v | S \cup T) \le f(v | S)$.
Thus, $T$ is not supermodular. 
Since there is no supermodular set with size larger 
 than $1$ and $f$ is not submodular, $\sm(f) = 1$.
\end{proof}

While the $\sa$ hierarchy does not subsume the $\mph$ hierarchy (see Proposition~\ref{prop:sa_much_larger_than_mph}), we show that there is a monotone set function in the lowest layer of the $\sa$ hierarchy (i.e.\ a subadditive function) and a notably high layer of the $\mph$ hierarchy.

\begin{proposition}
\label{prop:mph_much_larger_than_sa}
There exists a 
 monotone set function $f$ with $\sa(f) = 0$ and $\mph(f) = \frac{m}{2}$.
\end{proposition}
\begin{proof}
The proposition is a direct corollary of Proposition~L.2 in \cite{feige2015unifying}. In fact, consider a symmetric valuation $f$ over $[m]$, where $f([m]) = 2$, $f(\emptyset) = 0$, and $f(S) = 1$ otherwise. Clearly $f$ is subadditive so $\sa(f) = 0$. According to Corollary~F.5 of \cite{feige2015unifying}, $\mph(f) \ge \frac{m}{2}$. 
\end{proof}

\subsection{Further Comparisons between Hierarchies}

\begin{proposition}
\label{prop:sa_much_larger_than_sm}
There exists a monotone set function $f$ 
  with $\sm(f) = 1$ and $\sa(f) = m / 2$.
\end{proposition}
\begin{proof}
    Let $h_T(S) = \mathbb{I}[T \subseteq S]$. Consider function $f: 2^X \rightarrow \mathbb{R}^+$ where $X = [2t]$ and
    \[
        f(S) = \sum_{i \in [t]} h_{\{i, i + t\}}(S).
    \]
    $\sm(f) = 1$ because the only complement set to any item $i \in [t]$ is $i + t$. On the other hand, $T = \{t + 1, \dots, 2t\}$ is a complement set to $S = [t]$, so $\sa(f) = t = m / 2$.
\end{proof}

\begin{proposition}
\label{prop:sm_much_larger_than_sa}
There exists a monotone set function $f$ with
   $\sa(f) = 0$ and $\sm(f) = m - 1$.
\end{proposition}
\begin{proof}
Consider a symmetric $f: 2^X \rightarrow \mathbb{R}^+$, where $f(\emptyset) = 0$, $f(X) = 2$ and $f(S) = 1$ otherwise. $f$ is subadditive so $\sa(f) = 0$. On the other hand, $X \setminus \{u\}$ for any $u$ is a complement set to $u$, so $\sm(f) = m - 1$.
\end{proof}

\begin{proposition}
\label{prop:sa_much_larger_than_mph}
There exists a monotone set function $f$ 
  with $\mph(f) = 2$ and $\sm(f) = \sa(f) = m - 1$.
\end{proposition}
\begin{proof}
    Let $h_T(S) = \mathbb{I}[T \subseteq S]$. Consider function $f: 2^X \rightarrow \mathbb{R}^+$ where
    \[
        f(S) = \sum_{u \ne v} h_{\{u, v\}}(S).
    \]
    $f$ is in $\mph\text{-}2$ since its hypergraph representation consists of only hyperedges of size $2$. Now consider any $u$ and $T = X \setminus \{u\}$. For any $T' \subsetneq T$,
    \[
        f(u | T) = |T| > |T'| = f(u | T').
    \]
    In other words, $T$ is both supermodular and superadditive, and $\sm(f) = \sa(f) = m - 1$.
\end{proof}

\section{Expanding Approximation Guarantees for 
   Classic Maximization}

In this section, we focus on the connection between supermodular width
  and two classical optimization problems, namely,
   the constrained and welfare set-function maximization.
For submodular functions,
  greedy algorithms provide tight approximation guarantees for both problems
  \cite{nemhauser1978analysis,vondrak2008optimal}.
Here, we will show that simple modifications to
  these greedy algorithms can effectively 
  utilize the mathematical structure underlying the supermodular degree of $f$,
  namely the $\sm(f)$-scopic submodularity, for any set function $f$.
We prove that these extensions achieve approximation ratios 
  parametrized by the supermodular width with the same dependency 
  as the supermodular degree provides 
  \cite{feldman2014constrained,feige2013welfare} for 
  both maximization problems.
We complement our approximation results by 
  nearly tight information-theoretical lower bounds.

\subsection{Constrained Maximization}



We first focus on cardinality constrained maximization, 
  a problem at the center of resource 
  allocation and network influence 
  \cite{RichardsonDomingos,KKT,nemhauser1978analysis,vondrak2008optimal}.
Formally:

\begin{definition}[Cardinality Constrained Maximization]
Given a monotone set function $f: 2^X\rightarrow \Real^+\cup \{0\}$
 and integer $k > 0$, 
 compute a set $S\subseteq X$ with $|S| \leq k$ that maximizes
  $f(S)$.
\end{definition}


We will analyze an algorithm which performs {\em batched greedy selection},
  --- see Algorithm~\ref{alg:constrained_maximization} below ---
  where the batch size is a function of the supermodular width of $f$.
In particular, for an input set function,
  the batched greedy algorithm chooses a set of size not exceeding 
  $\sm(f) + 1$ 
   which maximizes marginal gain, till 
   all $k$ elements are chosen. 

\begin{algorithm}[t]
    let $S_0 \leftarrow \emptyset$; $i=0$\;
    \While {$|S_{i}|< k$} { 
        Let $i=i+1$; $T_i \leftarrow \argmax_{T' \subseteq [m], |T'| \le s} f(T' | S_i)$ where $s = \min\{\sm(f) + 1, k - |S_{i - 1}|\}$\;
        let $S_i \leftarrow S_{i - 1} \cup T_i$; \;
    }
    return $S^{\text{BatchedGreedy}} :=S_i$\;
\caption{Batched Greedy Selection for Constrained Maximization $(f,k)$}
\label{alg:constrained_maximization}
\end{algorithm}
Below, we show that this simple greedy algorithm
   provides strong approximation 
   guarantees in terms of the supermodular width of the input function.

\begin{theorem}[Extending \cite{feldman2014constrained}]
\label{thm:constrained_maximization}
For any monotone set function $f$ over $[m]$,
Algorithm~\ref{alg:constrained_maximization} achieves 
$\left(1 - e^{-1/(\sm(f) + 1)}\right)$-approximation
 for constrained maximization problem
  after making $O\left(m^{\sm(f) + 1}\right)$ value queries.
\end{theorem}
\begin{proof}
The proof uses similar ideas to those in \cite{feldman2014constrained}, which are originally from \cite{nemhauser1978analysis}.
Let $d = \sm(f)$ and
    (w.l.o.g.) let $S^* = [k] = \{1, \dots, k\}$ be an optimal solution.
    \begin{align}
        f(S^*) - f(S_i) & \le f(S^* \cup S_i) - f(S_i) \label{eq:monotone1} \\
        & \le f(S^* | S_i)\nonumber \\
        & = f([k] | S_i)  \nonumber \\
        & = \sum_{j \in [k]} f(j | [j - 1] \cup S_i) \nonumber \\
        & \le k \max_j f(j | [j - 1] \cup S_i) \nonumber \\
        & \le k \max_j \max_{U_j:U_j \subseteq [j - 1],\, |U_j| \le d }  f(j | U_j \cup S_i) \label{eq:dsubmodular1} \\
        & \le k \max_j \max_{U_j:U_j \subseteq [j - 1],\, |U_j| \le d } f(\{j\} \cup U_j | S_i) \label{eq:monotone2} \\
        & \le k f(S_{i + 1} | S_i)  \label{eq:algomax} \\
        & = k(f(S_{i + 1}) - f(S_i)) \nonumber
    \end{align}
    where \eqref{eq:monotone1} is by the monotonicity of $f$, \eqref{eq:dsubmodular1} is by the
    	equivalent $d$-scopic submodularity of $f$, \eqref{eq:monotone2} is again by the monotonicity of $f$, and 
    	\eqref{eq:algomax} is by the greedy property:
        $f(S_{i + 1} | S_i) = \max_{S: |S| \le d + 1} f(S | S_i)$.

    Now we have
    \begin{align*}
        f(S^*) - f(S_{i}) & \le \frac{k - 1}{k} (f(S^*) - f(S_{i - 1})) \\
        & \le \left(\frac{k - 1}{k}\right)^i (f(S^*) - f(S_0)) \\
        & = \left(\frac{k - 1}{k}\right)^i f(S^*) \\
        & \le e^{-i/k} f(S^*).
    \end{align*}
Because $f$ is monotone, we have $|T_i| = d+1$,
   for all intermediate steps, i.e., 
   $i < \lceil \frac{k}{\sm(f) + 1} \rceil$.
Thus, Algorithm~\ref{alg:constrained_maximization}
   takes exactly $t := \lceil \frac{k}{\sm(f) + 1} \rceil$ steps
   to terminate.
The function value of its output 
  $f(S^{\text{BatchedGreedy}}) := f(S_t) \ge 
  \left(1 - e^{-1 / (\sm(f) + 1)}\right) f(S^*)$.
\end{proof}

While in general, Theorem~\ref{thm:constrained_maximization} establishes a tighter approximation guarantee for the $\sm$ hierarchy than that for the $\sd$ hierarchy, we note that in case of submodular degree, if the positive dependency graph is given, the running times are often of the form $\mathrm{poly}(n) \cdot 2^{O(\sd(f))}$, which can be significantly better than $n^{O(\sm(f))}$ even if the submodular width $\sm(f)$ is much smaller than the submodular degree $\sd(f)$.

We now provide a nearly-matching information-theoretical lower bound, 
   suggesting that our approximation guarantee is essentially optimal.
In the theorem below, the exponent $k^{0.99}$ 
  can be replaced by any function of $k$ in $o(k)$.

\begin{theorem}\label{theo:CMLowerBound}
For any $d \in \mathbb{N}$, $\varepsilon > 0$, 
  and a large enough integer $k$, there exists a set 
  function $f: 2^{[m]} \rightarrow \mathbb{R}^+$, with $\sm(f) = d$, 
  such that any (possibly randomized) algorithm 
   that produces a $(1 / (d + 1) + \varepsilon)$-approximation
   (with a constant probability if randomized) 
   for the $k$-constrained maximization problem 
    makes at least 
    $\Omega\left((m / 2k)^{k^{0.99}}\right)$
  value queries.
\end{theorem}

\begin{proof}
The proof is based on similar high-level ideas to those in \cite{mirrokni2008tight}, but the detailed construction and key properties used are different.
Consider a ground set $X$ of $m$ elements, which contains
  a subset $R$ of $r$ ``special'' elements.
We will specify $r$ below.
We now construct a ``hard-to-distinguish'' function $f_R$ 
   such that for any  $S \subseteq X$, 
  $f_R(S) = g_R(|S|, \mathbb{I}[R \subseteq S])$
   for a function $g_R: \mathbb{N}\times \{0, 1\} \rightarrow \mathbb{R}$.
In other words, $f_R$ depends on the cardinality of $S$
  and whether or not $S$ completely contains $R$. 
For discussion below, let $D = d + 1$, and let $c_1$ and $c_2$ 
  be two integers to be determined later.
We set $|R| = r = c_1 \cdot D + 1$. 
We define $f_R$ as follows:
\[
        f_R(S) = \left\{\begin{array}{ll}
            \lfloor |S| / D \rfloor, & |S| \le c_1 D \\
            \lfloor (|S| - c_1 D) / D \rfloor + c_1, & c_1 D < |S| \le (c_1 + c_2) D ,\, R \not\subseteq S \\
            |S| - c_1(D - 1), & c_1 D < |S| \le (c_1 + c_2) D ,\, R \subseteq S \\
            |S| - (c_1 + c_2)(D - 1), & (c_1 + c_2) D < |S| \le (c_1 + c_2) D + c_2 (D - 1), R \not\subseteq S \\
            c_1 + c_2 D, & (c_1 + c_2) D < |S| \le (c_1 + c_2) D + c_2 (D - 1), R \subseteq S \\
            c_1 + c_2 D, & (c_1 + c_2) D + c_2 (D - 1) < |S| \le m
        \end{array}\right..
\]

We will use the following three properties of $f_R$:
\begin{itemize}
\item Whenever $|S| \mod D = D - 1$, for any $v \notin S$, $f_R(v | S) = 1$. 
Consequently, $\sm(f_R) \le d, \forall R \subseteq X$ with $|R| = r$. 
In fact, this property ensures that $f_R(v | S \cup T') \ge f_R(v | S \cup T)$,
  for any $v\in X$, $S, T\subseteq X$ with $|T| \ge D = d + 1$, 
  and any proper subset $T'$ of $T$ with $|S \cup T'| \mod D = D - 1$.
Note that such a subset $T'$ always exists.

\item $\max\left\{f_R(S) \mid |S| = (c_1 + c_2) D\right\} = c_1 + c_2 D$. 
  The maximum is achieved whenever $R \subseteq S$.
\item For any $S \subseteq X$ satisfying $|S| = c_1 + c_2 D$ 
  and $R \not\subseteq S$, $f_R(S) = c_1 + c_2$.
    \end{itemize}


First, consider $k = (c_1 + c_2) D$.
We have, for any $S$ with $|S| = k$:
\[
f_R(S) =\left\{ \begin{array}{ll}
c_1 + c_2 D & \text{if $R \subseteq S$}\\
c_1 + c_2 & \text{otherwise.}
\end{array}
 \right.
\]
Suppose $c_1 = o(c_2)$. 
To obtain an approximation ratio better than
        $(c_1 + c_2) / (c_1 + c_2 D) \rightarrow 1 / D$
     for $k$-constrained maximization of $f_R$, 
     any algorithm must find a set with size $k$ 
     that contains all special elements in $R$.

For our lower bound, we will analyze the following slightly relaxed variation of the problem:
Let $K = (c_1 + c_2) D + c_2 (D - 1) - 1 > k$.
Find a set of size $K$ which contains $R$ as a subset. 
Note that $K$ is the largest number where 
   $f_R(S)$ --- for $|S| = K$ --- depends on whether or not $S$ contains $R$. 
In this case, note that the algorithm has no incentive to 
   make queries of $f_R(S)$ for $|S| < K$ or $|S| > K$, 
   because the former reveals no more information than 
   querying any of its supersets of size $K$, and 
   the latter simply does not give any information.

We first focus on the query complexity of any deterministic optimization 
  algorithm. 
Assume the algorithm makes $T$ queries regarding
   $S_1, \dots, S_T$, where $|S_i| = K, \forall i\in [T]$,
  which are deterministically chosen when the algorithm is fixed.
We now establish a condition on $T$
  such that there is a subset $R$ such that 
  $R \not\subset S_i, \forall i\in [T]$.
Consider the distribution where the $r$ elements 
  are selected uniformly at random. 
Let $C_i$ be the event that $S_i$ contains $R$. 
Then, 
\[
\Pr[C_1 \cup \dots \cup C_T] 
  \le \sum_i \Pr[C_i] < \sum_i \left(\frac{|S_i|}{m}\right)^r = T \left(\frac{(c_1 + c_2) D + (D - 1) c_2 - 1}{m}\right)^{c_1 D + 1} \le T \left(\frac{2 c_2 D}{m}\right)^{c_1 D}.
    \]
So, if $T \le [m / (2 c_2 D)]^{c_1 D}$ then
   $\Pr[C_1 \cup \dots \cup C_T] < 1$. 
In other words, for any selections $S_1, \dots, S_T \subseteq X$ with $|S_i| = K$,
  there is a subset $R$, such that $R \not\subset S_i, \forall i\in [T]$,
  implying the deterministic algorithm
  with querying set $S_1, \dots, S_T$ will not find a good approximation to $f_R$.
Let $c_2 = \frac12 c_1^{1.01}$, 
   so $k^{0.99} = ((c_1 + c_2)D)^{0.99} \le (c_1^{1.01}D)^{0.99} \le c_1 D$. 
We have $\left(m / 2 c_2 D\right)^{c_1 D} \ge \left(m / 2k\right)^{k^{0.99}}$.
Thus, we conclude that any $(1 / (d + 1) + \varepsilon)$-approximation 
  deterministic algorithm must make 
  at least $(m / 2k)^{k^{0.99}}$ value queries.


Now consider a randomized optimization algorithm. 
Conditioned on the random bits of the algorithm, 
   the above argument still works. 
Taking expectation of the probability of success, 
   we see that the overall probability of success is at 
   most $T (2k / m)^{k^{0.99}}$. 
Thus, a constant probability of success requires 
  $T = \Omega\left((m / 2k)^{k^{0.99}}\right)$.
\end{proof}

\subsection{Welfare Maximization}


We now turn our attention to welfare maximization. 
Formally:

\begin{definition}[Welfare Maximization]
Given $n$ monotone set functions $f_1, \dots, f_n$ 
  over $2^{[m]}$, 
  compute $n$  disjoint sets $X_1, \dots, X_n$
  that maximizes $\sum_{i \in [n]} f_i(X_i)$.
\end{definition}

Because $f_1, \dots, f_n$ are monotone, the optimal solution to welfare
  maximization is a partition of $X=[m]$.
Thus, welfare maximization can also be viewed as a
   generalized clustering or multiway partitioning problem.

We will analyze the following greedy algorithm
 --- see Algorithm~\ref{alg:welfare_maximization} below ---
 which repeatedly assigns groups of elements to agents. 
At each step, the algorithm picks a set of 
  size not exceeding $\max_i \sm(f_i) + 1$ --- as opposed to one ---
  that provides the largest possible marginal 
  gain to some agent and assigns the set to that agent. 

\begin{algorithm}[t]
    \For {$j \in [n]$} {
        let $X_{j,0} \leftarrow \emptyset$\;
    }
   Let $d = \max_j \{\sm(f_j)\}$; let $i =0$\; 
\While {$\cup_j X_{j,i} \neq X$} {
 Let $i = i+1$; 
let $(T_i, j_i^*) = 
\argmax_{(T', j):\, |T'| \le s, j \in [n]} 
  f_j\left(T' | X_{j,i - 1}\right) \text{ where } s = \min\left\{d + 1, n - 
   \sum_j |X_{j,i - 1}|\right\}$\;
        Let $X_{j_i^*,i} \leftarrow X_{j_i^*,i - 1} \cup T_i$\;
        \For {$j \in [n] \setminus \{j_i^*\}$} {
            let $X_{j,i} \leftarrow X_{j,i - 1}$\;
        }
        return $X^{\text{BatchedGreedy}}_j := X_{j,i}$ for every agent $j$\;
    }
\caption{Batched Greedy for Welfare Maximization $(f_1, \dots, f_n)$}
\label{alg:welfare_maximization}
\end{algorithm}


We now prove the following approximation guarantee 
  in terms of supermodular width.
\begin{theorem}[Extending \cite{feige2013welfare}]
\label{thm:welfare_maximization}
For any collection of monotone set functions $f_1, \dots, f_n$ over $X = [m]$,
Algorithm~\ref{alg:welfare_maximization} achieves 
   $\frac{1}{2 + \max_i \{\sm(f_i)\}}$-approximation 
  for welfare maximization,
  after making $O\left(nm^{\max_i \{\sm(f_i)\} + 1}\right)$ value queries.
\end{theorem}
\begin{proof}
The proof uses similar ideas to those in \cite{feige2013welfare}, which are originally from \cite{fisher1978analysis}.
Following the notation in Algorithm~\ref{alg:welfare_maximization},
   we use $i$ to denote the step and $j$ to denote the agent's index.
Recall $d = \max_j \{\sm(f_j)\}$.
Suppose $(X_1^*, \dots, X_n^*)$ is an optimal solution to the welfare
  maximization of $(f_1, \dots, f_n)$.
Note that $\cup_j X_{j,i}$ is the subset of elements that has already been
   assigned at the end of step $i$.
Let $T_{j,i} =  X^*_{j}\setminus \cup_j X_{j,i}$ denote the set
 of elements of $X^*_{j}$ still available at the time.
Recall at step $i$, the set $T_i$ is allocated to agent $j_i^*$.
In other words, $X_{j_i^*,i} = X_{j_i^*,i - 1}\cup T_i$ and
  $f_{j_i^*}(X_{j^*_i,i}) - f_{j_i^*}(X_{j^*_i,i - 1}) = 
   f_{j_i^*}(T_i|X_{j^*_i,i - 1})$.
According to Algorithm~\ref{alg:welfare_maximization},  $|T_i|\leq d+1$.
We now prove the following
 instrumental inequality to our analysis.
\begin{eqnarray}\label{eqn:intru}
(d+2)\cdot \left(f_{j_i^*}(X_{j^*_i,i}) - f_{j_i^*}(X_{j^*_i,i - 1})\right) = 
(d + 2)\cdot f_{j_i^*}\left(T_i | X_{j^*_i,i - 1}\right) 
  \ge \sum_j 
\left(f_j(T_{j,i - 1} | X_{j,i - 1}) - f_j(T_{j,i} | X_{j,i})\right).
\end{eqnarray}


We divide the right hand terms according to two cases:

\noindent {\bf Case 1} (terms with $j\in [n] \setminus \{j_i^*\}$):
Note that $(T_{1,i-1}\cap T_i, \dots, T_{n,i-1}\cap T_i)$
  is a partition of $T_i$ because 
  $(X_1^*, \dots, X^*_n)$ is a partition of $X$.
Let $d_j = \left| T_{j,i-1}\cap T_i \right|$.
We have, 
\[
\sum_{j\neq j^*_i} d_j \leq |T_i|\leq d+1.
\]
Thus, for any $j \ne j^*$, 
   for analysis below, let's name the $d_j$
   elements in $T_{j,i - 1} \cap T_i$ as
   $\left\{u^{(j)}_{1}, \dots, u^{(j)}_{d_j}\right\}$.
Note that for $j \ne j^*$, $X_{j,i-1} = X_{j,i}$ and 
  $T_{j,i} = T_{j,i-1}\setminus \left\{u^{(j)}_{1}, \dots, u^{(j)}_{d_j}\right\}$, which implies
  the first equality below:
    \begin{align}
        & \hspace{-0.5in} f_j\left(T_{j,i - 1} | X_{j,i - 1}\right) - 
      f_j(T_{j,i} | X_{j,i}) \nonumber \\
        =\ & f_j\left(\left\{u^{(j)}_1, \dots, u^{(j)}_{d_j}\right\} | 
   T_{j,i} \cup X_{j,i - 1}\right) \nonumber \\
        =\ & \sum_{k=1}^{d_j} 
   f_j\left(u^{(j)}_k \big| \left\{u^{(j)}_1, \dots, u^{(j)}_{k - 1}\right\} 
   \cup T_{j,i} \cup X_{j,i - 1}\right) 
\nonumber \\
        \le\ & \sum_{k=1}^{d_j} 
    \left(\max_{V_k \subseteq \left\{u^{(j)}_1, \dots, u^{(j)}_{k - 1}\right\} \cup T_{j,i},\, |V_k| \le d} f_j\left(u^{(j)}_k | V_k \cup X_{j,i-1}\right) \right) \label{eq:welfare1} \\
        \le\ & \sum_{k=1}^{d_j}\left( \max_{V_k \subseteq \left\{u^{(j)}_1, \dots, u^{(j)}_{k - 1}\right\} \cup T_{j,i},\, |V_k| \le d} 
  f_j\left(\left\{u^{(j)}_k\right\} \cup V_k | X_{j,i - 1}\right)\right)
 \label{eq:welfare2} \\
        \le\ & d_j f_{j_i^*}(T_i | X_{j_i^*,i - 1}), \label{eq:welfare3} 
    \end{align}
where \eqref{eq:welfare1} follows from the $d$-scopic submodularity 
   of $f_j$ (note that $u^{(j)}_k \notin T_{j,i}$ for $j \ne j^*$), 
  \eqref{eq:welfare2} follows from monotonicity of $f_j$, 
   and \eqref{eq:welfare3} follows from the batched greedy selection
   of Algorithm~\ref{alg:welfare_maximization} that
   $f_{j_i^*}(T_i | X_{j_i^*,i - 1})$ 
   achieves the maximal possible marginal 
   among sets of size at most $d + 1$.
Summing over $j \ne j^*$, we have:
\begin{eqnarray} 
\sum_{j \ne j^*} \left(f_j\left(T_{j,i - 1} | X_{j,i - 1}\right) - 
   f_j\left(T_{j,i} | X_{j,i}\right)\right) 
  \le  \sum_{j \ne j^*} d_j \cdot f_{j_i^*}\left(T_i| X_{j_i^*,i - 1}\right)
 \le  (d + 1)\cdot f_{j_i^*}\left(T_i| X_{j_i^*,i - 1}\right) \label{eqn:non-receiver}
    \end{eqnarray}
\noindent {\bf Case 2} (term with $j_i^*$): 
\[
   f_{j_i^*}(T_{j_i^*,i - 1} | X_{j_i^*,i - 1}) + 
   f_{j_i^*}(X_{j_i^*,i - 1}) = f_{j_i^*}(T_{j_i^*,i - 1} \cup 
   X_{j_i^*,i - 1}) \le 
f_{j_i^*}(T_{j_i^*,i} \cup X_{j_i^*,i}) = 
  f_{j_i^*}(T_{j_i^*,i} | X_{j_i^*,i}) + f_{j_i^*}(X_{j_i^*,i}).
    \]
Therefore,
    \begin{equation} \label{eqn:receiver}
 f_{j_i^*}(T_{j_i^*,i - 1} | X_{j+i^*,i - 1}) - 
  f_{j_i^*}(T_{j_i^*,i} | X_{j_i^*,i}) \le 
  f_{j_i^*}(X_{j_i^*,i}) -    f_{j_i^*}(X_{j_i^*,i - 1})
 = f_{j_i^*}\left(T_i| X_{j_i^*,i - 1}\right).
    \end{equation}
Combining \eqref{eqn:non-receiver} and \eqref{eqn:receiver}, we have established \eqref{eqn:intru}.
Now, suppose the algorithm terminates after $t$ steps, 
 during which at step $i$, subset 
 $T_i$ is allocated to agent $j^*_i$. 
We have: 
    \begin{align*}
\sum_{j=1}^n f_j\left(X^*_j\right) 
        & = \sum_{j=1}^n f_j\left(T_{j,0} | X_{j,0}\right) \\
        & = \sum_{0 \le i < t} \sum_j 
       \left(f_j\left(T_{j,i} | X_{j,i}\right) - f_j\left(T_{j,i + 1} | 
               X_{j,i + 1}\right)\right) \\
        & \le (d + 2) \sum_{0 \le i < t} 
 \left(f_{j_i^*}(X_{j^*_i,i}) - f_{j_i^*}(X_{j^*_i,i - 1})\right) \\
        & = (d + 2) \sum_{j=1}^n f_j(X_{j,t})\\
& = (d+2) \sum_{j=1}^n f_j\left(X^{\text{BatchedGreedy}}_j\right).
    \end{align*}
\end{proof}

To show that our algorithm is nearly optimal,
  we prove the following information-theoretical lower bound:
Similar to Theorem \ref{theo:CMLowerBound}, 
  the exponent $(m/n)^{0.99}$ in the theorem below,
  can be replaced by any function of $m/n$ in $o(m/n)$.
\begin{theorem}\label{theo:WMLowerBound}
For any $d \in \mathbb{N}$, $\varepsilon > 0$, 
   there is a family of function 
   $f_1, \dots, f_n: 2^{[m]} \rightarrow \mathbb{R}^+$
   with $\sm(f_i) = d, \forall i\in [n]$, such that 
   any (possibly randomized) algorithm that produces 
   a $(1 / (d + 1) + \varepsilon)$-approximation 
  (with constant probability if randomized) for 
   the $n$-agent welfare maximization problem makes at 
   least 
  $\Omega\left((n / 2D)^{(m / n)^{0.99}}\right)$ 
   value queries.
\end{theorem}
\begin{proof}
This proof follows from a similar argument as 
  the proof for Theorem \ref{theo:CMLowerBound}.
Consider a ground set $X$ of $m$ elements, which contains
  a family of subsets $R_1, \dots, R_n$ of $r$ ``special'' elements.
We will specify $r$ below.
We construct a family of 
  (a slightly different version of)
   ``hard-to-distinguish'' set functions, which
 have the same supermodular degree.

To formulate these functions, let
  us first consider a set $R \subseteq X$
  of $r$ elements.
We construct a function $f_R$ 
  such that for any $S \subseteq X$, $f_R(S) = 
  g_R(|S|, \mathbb{I}[R \subseteq S])$,
  for a function $g_R: \mathbb{N}\times \mathbb{B} \rightarrow \mathbb{R}$.
Like in Theorem~\ref{theo:CMLowerBound}, 
  $f_R$ depends on the cardinality of $S$
  and whether or not $S$ completely contains $R$. 
In the discussion below, let $D = d + 1$, and 
   let $c_1$ and $c_2$ be functions of $m$ and $n$ to be determined later.
We set $|R| = r = c_1 D + 1$. 
We define $f_R$ as follows:
    \[
        f_R(S) = \left\{\begin{array}{ll}
            \lfloor |S| / D \rfloor, & R \not\subseteq S,\, |S| \le c_1 D + c_2 D^2 - 1 \\
            c_1 + c_2 D, & R \not\subseteq S,\, |S| \ge c_1 D + c_2 D^2 \\
            |S| - c_1(D - 1), & R \subseteq S,\, c_1 D < |S| \le (c_1 + c_2) D - 1 \\
            c_1 + c_2 D, & R \subseteq S,\, |S| \ge (c_1 + c_2) D
        \end{array}\right..
    \]

Like in Theorem \ref{theo:CMLowerBound} --- for similar reasons ---
   $f_R$ is in $\sm(f_R)\leq d$.

Each instance of the $n$-agent welfare maximization
  is defined by a family of subsets $R_1, \dots, R_n \subseteq X$
  satisfying for any $i \ne j$, $R_i \cap R_j = \emptyset$. 
The $i$-th agent's valuation function is then $f_i:= f_{R_i}$.

Consider the case $m = n \cdot s$ for $s := (c_1 + c_2) D$. 
We will use the following properties:
\begin{itemize}
\item A partition $(X_1, \dots, X_n)$ of $X$ is an optimal solution 
  to the $n$-agent welfare maximization problem with value functions
   $f_{R_1}, \dots, f_{R_n}$
  if and only if $X_i \supseteq R_i, \forall i\in [n]$.
The maximum welfare achievable is $n (c_1 + c_2 D)$.
        \item Let $t = c_1 D + c_2 D^2 - 1$, 
which is the largest size of $S$ such that $f(S)$ is not a constant. 
If no agent $i$ receives a set $X_i$ with $|X_i|\leq t$ 
  that is a superset of $R_i$, then the maximum possible welfare is $\lfloor ns / D \rfloor \le n(c_1 + c_2)$.
    \end{itemize}
So no algorithm can --- when $c_1 = o(c_2)$ ---
  achieve a better approximation ratio than
  $[n(c_1 + c_2)] / [n(c_1 + c_2 D)] \rightarrow 1 / D$
  without finding a set $X_i$ of size at most $t$ 
  containing $R_i$, for some $i$. 
We therefore reduce the analysis
  to a simple problem, where the goal is to find 
  a set of size $t = c_1 D + c_2 D^2 - 1$ containing some $R_i$ as a subset:
For each query, the algorithm can specify 
  a set $S$ and an index $i \in [k]$, and is informed --- by observing $f_i(S)$ --- 
  whether $S$ contains $R_i$. 
For similar reasons as in Theorem~\ref{theo:CMLowerBound},
   the algorithm has no incentive to make queries of 
  $f_i(S)$ for $|S| < t$ or $|S| > t$.

We focus on the query complexity of any deterministic welfare optimization
   algorithms. 
Assume the algorithm makes $T$ queries,
  $(S_1, k_1), \dots, (S_T, k_T)$ with $|S_i|  = t, \forall i\in [T]$,
  which are deterministically chosen when the algorithm is fixed.
We now establish a condition on $T$
  such that there is a family of disjoint subsets $(R_1,...,R_n)$ such that 
  $R_{k_i} \not\subset S_i, \forall i\in [T]$.
Consider the distribution $R_1, \dots, R_n$ uniformly at random 
  conditioned on $R_i \cap R_j = \emptyset$. 
Let $C_i$ be the event that $S_i$ contains $R_{k_i}$. 
Then,
    \[
        \Pr[C_1 \cup \dots \cup C_T] \le \sum_i \Pr[C_i] < \sum_i \left(\frac{|S_i|}{m}\right)^r = T \left(\frac{c_1 D + c_2 D^2 - 1}{m}\right)^{c_1 D + 1} \le T \left(\frac{2 c_2 D^2}{m}\right)^{c_1 D}.
    \]
So, if $T \le [m / (2 c_2 D^2)]^{c_1 D}$ then
   $\Pr[C_1 \cup \dots \cup C_T] < 1$. 
In other words, for any queries
   $(S_1, k_1), \dots, (S_T, k_T)$ with $|S_i|  = t, \forall i\in [T]$,
   there are disjoint subsets $R_1, \dots, R_n$ such that 
   $R_{k_i} \not\subset S_i, \forall i\in [T]$,
   implying the deterministic algorithm 
   with queries $(S_1, k_1), \dots, (S_T, k_T)$,
   will not find a good approximation to $(f_{R_1}, \dots, f_{R_n})$. 
Let $c_2 = \frac12 c_1^{1.01}$. 
We have
    \[
        \frac{m}{n} = s = (c_1 + c_2) D \ge c_2 D 
   = \frac{2 c_2 D^2}{2D} \Rightarrow \frac{m}{2 c_2 D^2} \ge \frac{n}{2D},
    \]
    and
    \[
        \left(\frac{m}{n}\right)^{0.99} = 
  s^{0.99} = (c_1 D + 1/2 c_1^{1.01} D)^{0.99} \le c_1 D.
    \]
So
 $\left(m / 2 c_2 D^2\right)^{c_1 D} \ge \left(n / 2D\right)^{(m / n)^{0.99}}$.
Thus, we conclude that any 
  $(1 / (d + 1) + \varepsilon)$-approximation algorithm  
   must make at least $(n / 2D)^{(m / n)^{0.99}}$ value queries.

Now consider a randomized welfare optimization algorithm. 
Conditioned on the random bits of the algorithm, 
  the above argument still works. 
Taking expectation of the probability of success, 
  we see that the overall probability of success 
 is at most $T (2D / n)^{(m / n)^{0.99}}$.
Thus, a constant probability of success
  requires $T = \Omega\left((n/ 2D)^{(m / n)^{0.99}}\right)$.
\end{proof}

\section{Efficiency of Simple Auctions}
\label{sec:efficiency}

In this section, we study the connection between the $\sa$ hierarchy and 
  efficiency of auctions.
We will draw extensively on previous work in this area,
  particularly on the characterization based on 
  the {\em CH hierarchy} --- see definition below --- 
   which is arguably the most simple class of set functions 
  with complementarity.


\begin{definition}[$d$-Constraint Homogeneous Functions \cite{feldman2016simple}]
\label{def:d-ch}
A set function $f$ over ground set $X$
   is $d$-constraint homogeneous (CH-$d$) 
   if there exists a value $\hat{f}$, and 
   disjoint sets $Q_1, \dots, Q_h \subseteq X$ 
   with $|Q_i| \leq d, \forall i\in [h]$, 
   such that 
 (1) $f(Q_i) = \hat{f}\cdot |Q_i|, \forall i\in [h]$, 
   and 
(2) the value of every set $S \subseteq [m]$ is 
    simply the sum of values of contained $Q_i$'s, i.e.,
        $f(S) = \sum_{Q_i \subseteq S} f(Q_i) = 
        \hat{f} \cdot \sum_{Q_i \subseteq S} |Q_i|.$
\end{definition}

We will show that previous characterization 
  of auction efficiency \cite{feldman2016simple}
  can be approximately extended from 
  the CH hierarchy to the $\sa$ hierarchy.


\subsection{Backgrounds: Related Definitions and Results}

We first restate a useful definition and a lemma for analyzing 
  the efficiency of auction mechanisms.

\begin{definition}[\cite{syrgkanis2013composable}]
An auction mechanism $\mathcal{M}$ 
  is $(\lambda, \mu)$-smooth for a 
  class of valuations $\mathcal{F} = \times_i \mathcal{F}_i$ 
  if for any valuation profile $f \in \mathcal{F}$, 
  there exists a (possibly randomized) 
  action profile $a_i^*(f)$ such that for every action profile $a$:
    \[
        \sum_i \mathbb{E}_{a_i' \sim a_i^*(f)}[u_i(a_i', a_{-i}; f_i)] \ge \lambda \cdot \mathrm{OPT}(f) - \mu \sum_i P_i(a),
    \]
   where $u_i(a_i'; f_i)$ is the utility of $i$ given action 
   profile $(a_i', a_{-i})$, $\mathrm{OPT}(f)$ is the optimum 
   social welfare given valuation profile $f$, 
   and $P_i(a)$ is the payment of $i$ given action profile $a$.
\end{definition}

\begin{lemma}[\cite{syrgkanis2013composable}]
If a mechanism is $(\lambda, \mu)$-smooth then 
 the price of anarchy w.r.t.\ coarse correlated equilibria 
 is at most $\max\{1, \mu\} / \lambda$.
\end{lemma}

For Single-bid Auction and Simultaneous Item First Price Auction (SIA), 
  we will derive our results from 
  the following results for CH-$d$ and $\mph$-$d$.

\begin{theorem}[Smoothness of Single-bid Auction with CH-$d$ Valuations \cite{feldman2016simple}]
\label{lem:single-bid_with_ch}
    Single-bid Auction is a $((1 - e^{-d}) / d, 1)$-smooth mechanism when agents have CH-$d$ valuations.
Consequently, Single-bid Auction has a PoA of 
  $(1 - e^{-d}) / d$ with CH-$d$ valuations 
  w.r.t.\ coarse correlated equilibria.
\end{theorem}

\begin{theorem}[Smoothness of SIA with $\mph$-$d$ Valuations \cite{feige2015unifying}]
\label{lem:sia_with_mph}
    For SIA, when bidders have MPH-$d$ valuations, both the correlated price of anarchy and the Bayes-Nash price of anarchy are at most $2d$. The bound follows from a smoothness argument.
\end{theorem}

A key concept to extend these results to other valuation
  classes is the following notion of  pointwise approximation defined 
   in \cite{devanur2015simple}.

\begin{definition}[Pointwise Approximation \cite{devanur2015simple}]
\label{def:pointwse_approximation}
A class of set functions $\mathcal{F}$ over ground set $X$ 
  is pointwise $\beta$-approximated by another class $\mathcal{F}'$
  of set functions over $X$ if $\forall f \in \mathcal{F}, S \subseteq X$, 
  $\exists f'_S \in \mathcal{F}'$ such that 
  (1) $\beta f'_S(S) \ge f(S)$ and (2) $\forall T \subseteq X$, 
  $f'_S(T) \le f(T)$.
\end{definition}

For example:

\begin{proposition}[\cite{feldman2016simple}]
\label{lem:maxf_by_f}
    The class $\max(\mathcal{F})$ is pointwise $1$-approximated by the class $\mathcal{F}$.
\end{proposition}

We say a function $f': 2^X \rightarrow \mathbb{R}$ pointwise $\beta$-approximates $f: 2^X \rightarrow \mathbb{R}$ (at $X$), if (1) $\beta f'(X) \ge f(X)$, and (2) $\forall T \subseteq X$, 
  $f'(T) \le f(T)$.

The following lemma of 
   \cite{devanur2015simple}
  provides a way to translate PoA bounds 
  between classes via pointwise approximation. 

\begin{lemma}[Extension Lemma \cite{devanur2015simple}]
\label{lem:poa_from_approximation}
    If a mechanism for a combinatorial auction setting is $(\lambda, \mu)$-smooth for the class of set functions $\mathcal{F}'$, and $\mathcal{F}$ is pointwise $\beta$-approximated by $\mathcal{F}'$, then it is $\left(\frac\lambda\beta, \mu\right)$-smooth for the class $\mathcal{F}$. And as a result, if a mechanism for a combinatorial auction setting has a PoA of $\alpha$ given by a smoothness argument for the class $\mathcal{F}'$, and $\mathcal{F}$ is pointwise $\beta$-approximated by $\mathcal{F}'$, then it has a PoA of $\alpha\beta$ for the class $\mathcal{F}$.
\end{lemma}

\subsection{Efficiency of Simple Auctions Parametrized by $\sa$}

Applying Lemma~\ref{lem:poa_from_approximation}, 
  we are able to translate Theorems~\ref{lem:single-bid_with_ch}~and~\ref{lem:sia_with_mph}
  to the $\sa$ hierarchy.

\begin{theorem}[Efficiency of Single-bid Auction with $\sa$-$d$ Valuations]
\label{thm:single-bid_with_sa}
When agents have valuations $f_1, \dots, f_n \in \max(\sa\text{-}d)$, 
  Single-bid Auction has a price of anarchy of 
  at most $\frac{2d}{1 - e^{-2d}}\cdot H_\frac{m}{2d}$ 
  w.r.t.\ coarse correlated equilibria. 
\end{theorem}


\begin{theorem}[Efficiency of SIA with $\sa$-$d$ Valuations]
\label{thm:sia_with_sa}
    When agents have valuations $f_1, \dots, f_n \in \max(\sa\text{-}d)$, SIA has a price of anarchy of at most $8d\cdot H_\frac{m}{2d}$ w.r.t.\ coarse correlated equilibria.
\end{theorem}

Formally, 
  Theorems~\ref{thm:single-bid_with_sa}~and~\ref{thm:sia_with_sa} follow from
  Theorems~\ref{lem:single-bid_with_ch} 
  and~\ref{lem:sia_with_mph} respectively, with the help of 
  Lemma~\ref{lem:poa_from_approximation},
  Proposition~\ref{lem:maxf_by_f},
  and the technical lemma (Lemma~\ref{lem:sa_by_ch}) that we will establish below, showing that 
  for any $d\in {\mathbb N}$, functions in $\sa\text{-}d$ 
  can be approximated by CH-$2d$ functions.
In particular, Lemma~\ref{lem:sa_by_ch} establishes 
 the approximation of $\sa$ hierarchy by CH
 hierarchy with a loss of factor $O(\log m)$. 




\begin{lemma}[Pointwise Approximation of $\sa$ Hierarchy by CH-Hierarchy]
\label{lem:sa_by_ch}
For any $d \in {\mathbb N}$, 
  $\sa$-$d$ is pointwise $2 H_\frac{m}{2d}$-approximated 
  by CH-$2d$,
  where $H_i = \sum_{k \in [i]} \frac{1}{k}$ 
  is the $i$-th harmonic number.
\end{lemma}
\begin{proof}
Our proof is inspired by the constructions of 
  \cite{devanur2015simple} and \cite{feldman2016simple}.

For any $f \in \sa$-$d$ over $X = [m]$, we first apply 
  the following greedy construction to obtain a partition
   $\mathcal{Q} = \{Q_i\}_{i \in [q]}$ of $[m]$ into 
   sets of size not exceeding $2d$: 
   At step $i$, we select a new set 
    $Q_i \subseteq [m] \setminus (Q_1 \cup \dots \cup Q_{i - 1})$, 
    with maximum $f(Q_i)$, among all sets of size at most $2d$.

We first prove by contradiction that there exists a function $g$ in CH-$2d$ which $2H_{\frac{m}{2d}}$-approximates $f$ at $[m]$.
 That is, (1) $2H_{\frac{m}{2d}} g([m]) \ge f([m])$ and 
  (2) $\forall T \subseteq [m]$, $g(T) \le f(T)$.

Suppose this statement is not true.
Let
\[
h_\mathcal{Q}(T)=
   \frac{f([m])}{\beta\cdot |\cup_i Q_i|}\sum_{Q_i\subseteq T}|Q_i|.
\]
Note that $h_\mathcal{Q}\in$ CH-$2d$ because 
  $|Q_i| \le 2d, \forall Q_i \in \mathcal{Q}$.
We now construct a series of functions based on $h_\mathcal{Q}$, and 
  prove that for any $\beta > 0$, if there is no $g$
  among these functions that is a $\beta$-approximation of $f$ at $[m]$ ---
  that is, there is no $g$ such that
  (1) $\beta g([m]) \ge f([m])$ and (2)
   $\forall T \subseteq [m]$, $g(T) \le f(T)$,
(below we will refer to this condition as Assumption (*)) ---
  then $\beta < 2H_{\frac{m}{2d}}$.

First consider $h_\mathcal{Q}$.
Note that $\beta h_\mathcal{Q}([m]) = \beta \frac{f([m])}{\beta} \ge f([m])$,
  because $\mathcal{Q}$ is a partition of $[m]$.
Assumption (*) then implies
there is a $T_1$ such that $h_\mathcal{Q}(T_1) > f(T_1)$. 
W.l.o.g. assume $T_1$ is a union of sets from 
  $\mathcal{Q}$ (such $T_1$ exists because $f$ is monotone).

Let $S_1 = [m]$. 
We now iteratively define $S_i = S_{i - 1} \setminus T_{i - 1}$, 
  and construct its associated $T_i$.
The construction maintains the following invariant:
Both $S_{i}$ and $T_{i}$ are unions of sets from $\mathcal{Q}$.
The former follows directly from 
  the iterative property that $S_{i - 1}$ and $T_{i - 1}$ are 
  both unions of sets from $\mathcal{Q}$.
Our construction below will ensure the latter.

Let $\mathcal{Q}_{S_i} = \{Q \in \mathcal{Q} \mid Q \subseteq S_i\}$.
Let 
\[h_{\mathcal{Q}_{S_i}} = \frac{f([m])}{\beta\cdot |\cup_{j: Q_j \in \mathcal{Q}_{S_i}} Q_j|} \sum_{j: Q_j \in \mathcal{Q}_{S_i}} |Q_j|.\]
Again, 
 $h_{\mathcal{Q}_{S_i}} \in$ CH-$2d$, and 
 $h_{\mathcal{Q}_{S_i}}([m]) = \frac{f([m])}{\beta}$. 
Assumption (*) then implies
  there is a $T_i$ such that $h_{\mathcal{Q}_{S_i}}(T_i) > f(T_i)$.
Again, w.l.o.g.\ assume $T_i$ is a union of sets from 
  $\mathcal{Q}$ (such $T_i$ exists because $f$ is monotone).
This iterative process terminates, producing a partition 
$\{T_i\}_{i \in [t]}$ of $[m]$, which satisfies:
    \[
        \sum_i f(T_i) < \sum_i h_{\mathcal{Q}_{S_i}}(T_i) = \frac{f([m])}{\beta} \sum_i \frac{|T_i|}{|S_i|} \le \frac{f([m])}{\beta} \sum_{i \in [t]} \frac{1}{i} \le \frac{f([m])}{\beta} H_\frac{m}{2d}.
    \]

We now show that $\sum_i f(T_i) \ge \frac{1}{2}f([m])$. 
Recall that each member in partition $\{T_i\}_i$ is a unions of sets from 
  $\mathcal{Q}$. We renumber $\{T_i\}_i$, in a way that for any $i < j$, there is some $T_i \supseteq Q_k \in \mathcal{Q}$, such that for any $T_j \supseteq Q_l \in \mathcal{Q}$, $k < l$. That is, the smallest index $k$ where $Q_k \in T_i$ is smaller than the smallest index $l$ where $Q_l \in T_j$, as long as $i < j$.

Since $(T_1, \dots, T_t)$ is a partition of $[m]$, we have:
    \begin{align}
        f([m]) & = \sum_i f(T_i | T_{i + 1} \cup \dots \cup T_t) \nonumber \\
        & \le \sum_i \max \{f(T_i | U_i) \mid U_i \subseteq T_{i + 1} \cup \dots \cup T_t,\, |U_i| \le d\} \label{eq:sach1} \\
        & \le \sum_i \max \{f(T_i \cup U_i) \mid U_i \subseteq T_{i + 1} \cup \dots \cup T_t,\, |U_i| \le d\} \label{eq:sach2} \\
        & = \sum_i \max \{(f(U_i | T_i) + f(T_i)) \mid U_i \subseteq T_{i + 1} \cup \dots \cup T_t,\, |U_i| \le d\} \nonumber \\
        & \le \sum_i \max \{(f(U_i | V_i) + f(T_i)) \mid U_i \subseteq T_{i + 1} \cup \dots \cup T_t,\, |U_i| \le d,\, V_i \subseteq T_i,\, |V_i| \le d\} \label{eq:sach3} \\
        & \le \sum_i \max \{(f(U_i \cup V_i) + f(T_i)) \mid U_i \subseteq T_{i + 1} \cup \dots \cup T_t,\, |U_i| \le d,\, V_i \subseteq T_i,\, |V_i| \le d\} \label{eq:sach4} \\
        & \le \sum_i (f(Q_{k_i}) + f(T_i)),\, \text{where } k_i = min \{k \mid T_i \supseteq Q_{k} \in \mathcal{Q}\} \label{eq:sach5} \\
        & \le \sum_i 2f(T_i), \label{eq:sach6}
    \end{align}
where \eqref{eq:sach1} and \eqref{eq:sach3} 
  follow from $d$-scopic subadditivity of $f$,
   \eqref{eq:sach2}, \eqref{eq:sach4} and \eqref{eq:sach6} 
  follow from monotonicity of $f$, and
   \eqref{eq:sach5} holds because, according to the construction of 
   $\{Q_l\}_l$, $Q_{k_i}$ maximizes $f$ among all 
    sets of size $2d$ contained in 
    $Q_{k_i} \cup \dots \cup Q_q \supseteq T_i \cup \dots \cup T_t$, 
   and in particular $U_i \cup V_i \subseteq T_i \cup \dots \cup T_t$.

Consequently, it follows from $\sum_i f(T_i) \ge \frac{1}{2}f([m])$ that:
    \[
        \frac{H_\frac{m}{2d} f([m])}{\beta} > \sum_i f(T_i) \ge \frac12 f([m]) \Rightarrow \beta < 2 H_\frac{m}{2d}.
    \]
Thus, Assumption (*) with $\beta \ge 2 H_\frac{m}{2d}$
   leads to a contradiction. 
Therefore, we have established that there exists a CH-$2d$ function $g$
 such that (1) $g([m]) \ge 2H_{\frac{m}{2d}} f([m])$ and 
  (2) $\forall T \subseteq [m]$, $g(T) \le f(T)$.

As in~\cite{feldman2016simple}, the above proof can be simply extended
  to prove for any $S \subseteq X$, 
  there exists a CH-$2d$ function $g$
  such that (1) $g(S) \ge 2H_{\frac{m}{2d}} f([m])$ and 
  (2) $\forall T \subseteq [m]$, $g(T) \le f(T)$.
Essentially, we restrict the function $f$ to $2^S$, 
  apply the argument above, and then span the obtained function back to $2^{X}$.

Therefore, $\sa$-$d$ is pointwise $2 H_\frac{m}{2d}$-approximated
  by CH-$2d$.  
\end{proof}



We further analyze previously known hard instances 
  to both auctions, and show that they provide 
  almost matching lower bounds to the above two efficiency upper bounds.

\begin{theorem}
There is an instance with $\sa\text{-}d$ valuations 
  for any $d$, where the PoS of Single-bid Auction 
  is at least $d + 1 - \varepsilon / d$ for any $\varepsilon > 0$.
\end{theorem}
\begin{proof}
    Consider two players with valuations $f_1$ and $f_2$ over ground set $X = [m]$. Let $h_T(S) = \mathbb{I}[T \subseteq S]$, $f_1(S) = \sum_{2 \le i \le d + 1} h_{\{1, i\}}$, and $f_2(S) = \mathbb{I}[1 \in S] \left(\frac{d}{d + 1} + \varepsilon \right)$. Both $f_1$ and $f_2$ are in $\sa\text{-}d$ because there are at most $d + 1$ items which matter to the valuations. As shown in Proposition~3.9 of \cite{feldman2016simple}, Single-bid Auction has a PoS of $d + 1 - \varepsilon / d$ on this instance.
\end{proof}

\begin{theorem}
    There is an instance with $\sa\text{-}d$ valuations for any $d$, where the PoA of SIA is at least $d + 1 / (d + 1)$.
\end{theorem}
\begin{proof}
    Consider the instance given in Theorem~2.5 of \cite{feige2015unifying}. That is, a projective plane of order $d + 1$. There are $d(d + 1) + 1$ players, each desiring only a bundle of size $d + 1$, so the valuations of all players are in $\sa\text{-}d$. As shown by Theorem~2.5 of \cite{feige2015unifying} SIA has a PoA of at least $d + 1 / (d + 1)$ on the above instance.
\end{proof}

\subsection{Efficiency of Simple Auctions Parametrized by $\sm$}
\label{sec:efficiency_with_sm}

As a byproduct of our efficiency results 
  for the $\sa$ hierarchy, we prove similar,
  but slightly weaker, results for the $\sm$ hierarchy. 
We note that these bounds extend
  a central result in \cite{feldman2016simple}, which states that when agents have valuations in $\max(\sd\text{-}d \cap \supadd)$, Single-bid Auction has a PoA of $O(d^2 \log m)$.

\begin{theorem}[Extending \cite{feldman2016simple}]
\label{thm:single-bid_with_sm}
 When agents have valuations 
  $f_1, \dots, f_n \in \max(\sm\text{-}d \cap \supadd)$, 
  Single-bid Auction has a price of anarchy of 
  at most $\frac{(d + 1)^2}{1 - e^{-(d + 1)}}\cdot H_\frac{m}{d + 1}$ 
  w.r.t.\ coarse correlated equilibria.
\end{theorem}

\begin{theorem}
\label{thm:sia_with_sm}
When agents have valuations 
  $f_1, \dots, f_n \in \max(\sm\text{-}d \cap \supadd)$, 
  SIA has a price of anarchy of at most 
  $2(d + 1)^2\cdot H_\frac{m}{d + 1}$ w.r.t.\ coarse correlated equilibria.
\end{theorem}


Like Theorems~\ref{thm:single-bid_with_sa}~and~\ref{thm:sia_with_sa},
  the two theorems above follow from
  Theorems~\ref{lem:single-bid_with_ch}~and~\ref{lem:sia_with_mph} respectively, with the help of
  Lemma~\ref{lem:poa_from_approximation},
  Proposition~\ref{lem:maxf_by_f},
  and the technical lemma below.

\begin{lemma}
\label{lem:sm_by_ch}
For any $d \in {\mathbb N}$,  $\sm\text{-}d \cap \supadd$
  is pointwise $(d + 1) H_\frac{m}{d + 1}$-approximated 
  by CH-$(d + 1)$.
\end{lemma}
\begin{proof}
The proof essentially follows from the 
  same argument as that for Lemma~\ref{lem:sa_by_ch}.
For any superadditive $f$ over $X=[m]$ with $\sm(f) \le d$, 
  we first greedily construct a partition $\{Q_i\}_i$ of $X$:
  At step $i$, 
  we select a set $Q_i$ of at most $d + 1$ elements from 
  $X \setminus (Q_1 \cup \dots \cup Q_{i - 1})$
  that maximizes $f(Q_i)$.
For $x\in X$, let $\text{index}(x) = i$  iff $x\in Q_i$.
W.l.o.g., for analysis below, 
   we assume that elements in $X$ are already sorted 
  (or are renumbered) according to their indices, 
  i.e., if $\text{index}(x) < \text{index}(y)$
  then $x < y$.

Following the proof of Lemma~\ref{lem:sa_by_ch}, 
 we focus on proving by contradiction that there exists a CH-$(d+1)$ 
 function $g$ that $(d + 1) H_\frac{m}{d + 1}$-approximates $f$.
 That is (1) $(d + 1) H_\frac{m}{d + 1} g([m]) \ge f([m])$ and 
 (2) $\forall T \subseteq [m]$, $g(T) \le f(T)$.

Suppose this statement is not true.
Letting
\[h_{\mathcal{Q}_{S_i}} = \frac{f([m])}{\beta\cdot |\cup_{j: Q_j \in \mathcal{Q}_{S_i}} Q_j|} \sum_{j: Q_j \in \mathcal{Q}_{S_i}} |Q_j|,\]
and starting with $S_1 =[m]$, we can use the same 
  iterative process to construct a sequence
  $((S_1, T_1), \dots, (S_t,T_t))$ for some $t\in {\mathbb N}$
  such that (1) for all $i\in [t]$,
   both $S_i$ and $T_i$ are unions of sets from $\mathcal{Q}$,
  (2) $(T_1, \dots, T_n)$ is a partition of $X$,
   and (3)
$\sum_i f(T_i) < \frac{f([m])}{\beta} H_\frac{m}{d + 1}$,
  under that assumption that $\beta$ is a parameter 
  such that all induced functions (from CH-$(d + 1)$) satisfying
  (1) $h_{\mathcal{Q}_{S_i}}([m]) = \frac{f([m])}{\beta}$,
  and (2) $h_{\mathcal{Q}_{S_i}}(T_i) > f(T_i)$.

Now we have:
    \begin{align}
        f([m]) & = \sum_k f(k | [k - 1]) \nonumber \\
        & \le \sum_k \max_{U_k:\, U_k \subseteq [k - 1],\, |U_k| \le d} f(k | U_k) \label{eq:smch1} \\
        & \le \sum_k \max_{U_k:\, U_k \subseteq [k - 1],\, |U_k| \le d} f(\{k\} \cup U_k) \label{eq:smch2} \\
        & \le \sum_k f(Q_{\lceil\frac{m - k + 1}{d + 1}\rceil}) \label{eq:smch3} \\
        & = \sum_j |Q_j| f(Q_j) \nonumber \\
        & \le (d + 1) \sum_j f(Q_j) \nonumber 
    \end{align}
where \eqref{eq:smch1} follows from the fact that 
   $f \in \sa\text{-}d$, \eqref{eq:smch2} follows from monotonicity,
 and \eqref{eq:smch3} holds because
   by the construction of $\{Q_i\}_i$, 
  $Q_{\lceil\frac{m - k + 1}{d + 1}\rceil}$ 
  maximizes $f$ among all sets of size $d + 1$ contained in $[k]$.

On the other hand, since every $T_i$ is a union of some $Q_j$'s, according to superadditivity of $f$,
    \begin{align*}
        \beta & < \left(H_\frac{m}{d + 1} f([m])\right) \left(\sum_i f(T_i)\right)^{-1} \\
        & \le \left(H_\frac{m}{d + 1} f([m])\right) \left(\sum_i f(Q_i)\right)^{-1} \\
        & \le \left(H_\frac{m}{d + 1} (d + 1) \sum_i f(Q_i)\right) \left(\sum_i f(Q_i)\right)^{-1} \\
        & = (d + 1) H_\frac{m}{d + 1}.
    \end{align*}

For all other set $S \subseteq X$, 
  we can apply a similar restricting-and-spanning-back
  argument with the above construction to  prove
  that there exists a CH-$(d+1)$ function $g$
  such that (1) $(d + 1) H_\frac{m}{d + 1} g(S) \ge f([S])$, and 
  (2) $\forall T \subseteq [m]$, $g(T) \le f(T)$.
\end{proof}

\section{Remarks}

\subsection{Further Comparative Analysis}

As observed by Eden {\em et al.}\ \cite{eden2017simple}, the right measure
  of complementarity often varies from application to application.
This seems to be true even with the supermodular vs superadditive widths.
We note that while the $\sd$ and $\sm$ hierarchies 
  give nontrivial bounds on the PoA of 
  simple auctions, $\sa$ hierarchy seems to capture 
  the intrinsic property needed by efficiency guarantees for simple auctions. 
It provides tighter characterization of PoA 
  with a gap of $\log m$ (instead of $d \log m$) between upper and lower
  bounds.
On the other hand, while $\sm$ hierarchy captures
  the intrinsic property needed by the constrained/welfare maximization, 
  it remains open whether 
  a small superadditive width provides any approximation guarantee
  for the two optimization problems.

The $\mph$ hierarchy takes a 
  different approach from ours --- it relies
  on a syntactic definition which provides elegant and intuitive 
  structures. 
In contrast, 
  both $\sm$ and $\sa$ hierarchies ---  like the $\sd$ hierarchy
  before it ---
  are built on concrete 
  natural concepts of witnesses and semantic intuition of complementarity.
In the current definition, the MPH hierarchy is not 
  an extension to submodularity or subadditivity.
Rather --- as shown in \cite{feige2015unifying} ---
  MPH can be considered as an extension to the fractionally 
  subadditive (or XOS) class proposed in \cite{lehmann2006combinatorial}.
We therefore consider $\sm$, $\mph$ and $\sa$ parallel measures of complementarity,
  just like submodularity, fractional subadditivity and subadditivity in
  the complement-free case.
One key difference is that the three hierarchies seem to diverge at higher
  levels of complementarity, as opposed to the fact that submodular functions
  are all fractionally subadditive, and fractionally subadditive functions
  are all subadditive.
This phenomenon provides further evidence that the three hierarchies are likely to
  capture different aspects of complementarity.
See Figure \ref{fig:relationship} for a comparison.

\begin{figure}[t]
    \centering
    \includegraphics[width=\linewidth]{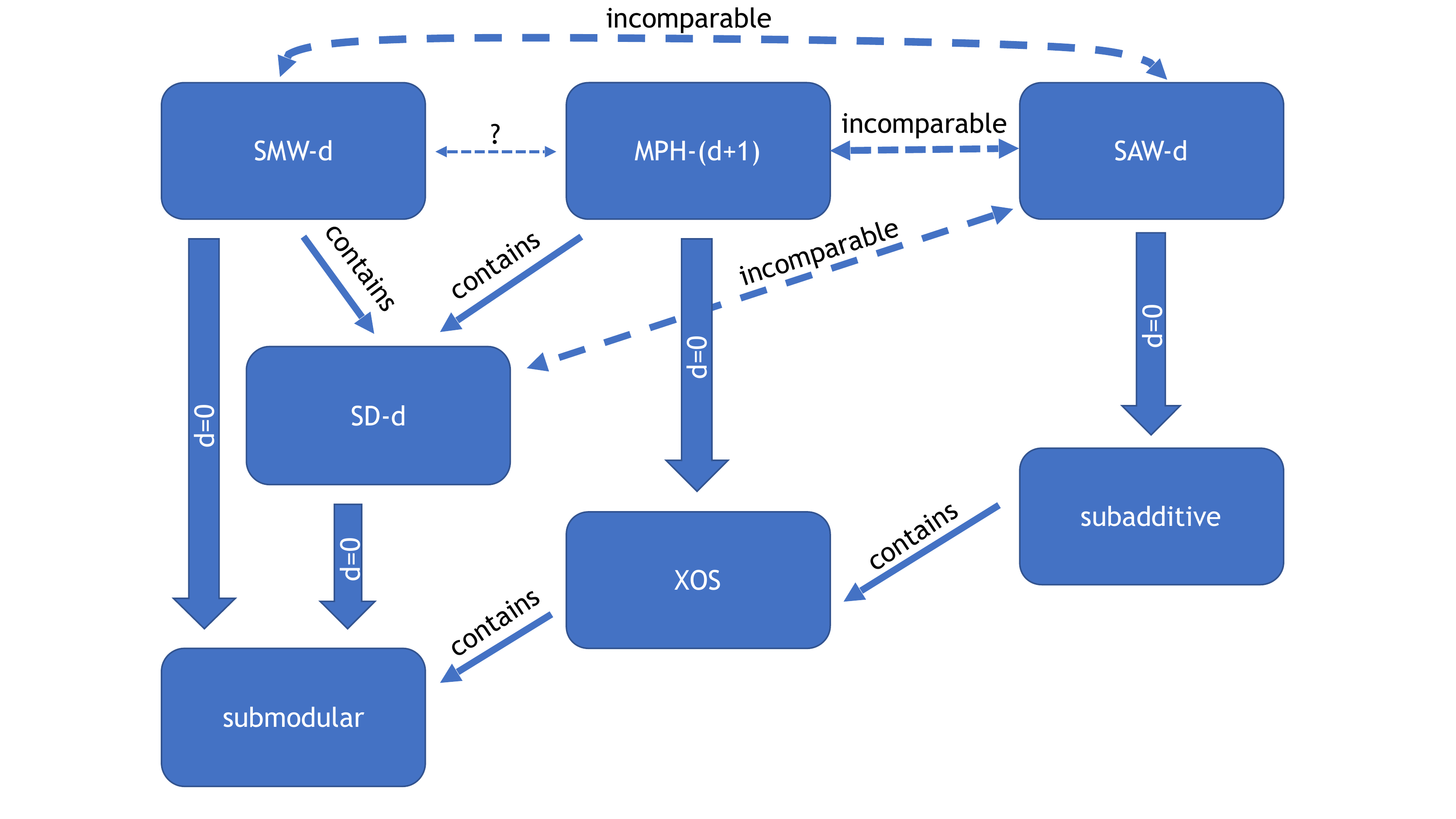}
    \caption{Relationship between hierarchies.}
    \label{fig:relationship}
\end{figure}

We also note that all upper bounds supported by our hierarchies
  are accompanied by almost matching lower bounds, which we consider as a 
  justification of our definitions --- they manage to categorize
  set functions roughly according to their ``hardness'' in different
  settings (i.e.\ optimization for $\sm$ and efficiency for $\sa$).
In contrast, while the less inclusive supermodular degree hierarchy
  supports a number of upper bounds, to our knowledge, none of those
  results are proven tight.

\subsection{Final Remarks and Open Problems}


Our SMW and SAW hierarchies may be applied to other problem settings.
For example, for the online secretary problem based on supermodular degree \cite{feldman2017building}, 
	we believe that with a slight modification of the algorithms and the analysis, 
	we could replace
	supermodular degree with supermodular width as well for this problem;
    also, $\sm$-$d$ functions are efficiently PAC-learnable under product distributions \cite{zhang2019learning}.
It may be possible to look into other venues where SMW and SAW hierarchies are applicable.


There are also a few technical questions to be answered:
\begin{itemize}
\item 
Does $\mph$-$(d + 1)$ --- which subsumes $\sd\text{-}d$  ---  
  include all $\sm\text{-}d$ functions?

\item Can we improve the SAW-based efficiency characterization of 
  of Single-bid Auction and SIA to $O(d)$?

\item Can the MPH hierarchy be used to characterize
      constrained set function maximization? 
\end{itemize}

\paragraph{Acknowledgements.} We thank Vincent Conitzer for helpful feedback and discussion, and anonymous reviewers for their insightful comments and suggestions.


	

\bibliographystyle{plain}
\bibliography{biblio}



\end{document}